\documentclass{article}

\usepackage{todonotes,amsmath,amsthm, amssymb}
\usepackage{enumitem, tablefootnote}
\usepackage[linesnumbered,boxed,vlined,ruled]{algorithm2e}
\usepackage[paperwidth=8.5in,paperheight=11in,margin=1in]{geometry}

\newtheorem{observation}{Observation}
\newtheorem{definition}{Definition}

\newtheorem{lemma}{Lemma}
\newtheorem{theorem}{Theorem}
\newtheorem{corollary}{Corollary}
\newtheorem{example}{Example}

\renewcommand{\b}{\mathbf{b}}
\newcommand{\M}{{M}}
\newcommand{\T}{\mathcal{T}}
\newcommand{\vgraph}{OSP-graph}
\newcommand{\ver}{{\cal O}_i^\T}
\newcommand{\verp}{{\cal O}_i^{\T'}}

\newcommand{\vect}[1]{\mathbf{#1} }
\renewcommand{\a}{\vect a}

\newcommand{\bi}{\b_{-i}}
\renewcommand{\xi}{\vect{x}_{-i}}

\newcommand{\x}{\vect{x}}

\renewcommand{\c}{\vect{c}}

\newcommand{\always}[1]{\text{$#1$-always}}
\newcommand{\sometime}[1]{\text{$#1$-sometime}}
\newcommand{\ti}{\vect t_{-i}}
\newcommand{\bb}[1]{\vect{b}^{(#1)}}
\newcommand{\bbb}[1]{b^{(#1)}_i}
\newcommand{\ccc}[1]{c^{(#1)}_i}
\renewcommand{\a}{\vect a}
\newcommand{\ab}[1]{\a^{(#1)}}
\newcommand{\abb}[1]{a^{(#1)}_i}
\newcommand{\LL}[1]{L^{(#1)}}
\newcommand{\out}{\mathcal S}

\newcommand{\neither}{unclear}
\newcommand{\sel}{\mathrm{sel}}

\DeclareMathOperator*{\med}{med}

\begin{document}

\title{Two-way Greedy: Algorithms for Imperfect Rationality} 

\author{Diodato Ferraioli\thanks{Universit\`a di Salerno, Italy. Email: {\tt dferraioli@unisa.it}} \and Paolo Penna\thanks{ETH Zurich, Switzerland. Email: {\tt paolo.penna@inf.ethz.ch}} \and Carmine Ventre\thanks{King's College London, UK. Email: {\tt carmine.ventre@kcl.ac.uk}}}


\date{}

\maketitle

\thispagestyle{empty}

\begin{abstract}
The realization that selfish interests need to be accounted for in the design of algorithms has produced many interesting and valuable contributions in computer science under the general umbrella of algorithmic mechanism design. Novel algorithmic properties and paradigms have been identified and studied in the literature. Our work stems from the observation that selfishness is different from rationality; agents will attempt to strategize whenever they perceive it to be convenient according to their imperfect rationality. Recent work in economics \cite{liosp}  has focused on a particular notion of imperfect rationality, namely absence of contingent reasoning skills, and defined obvious strategyproofness (OSP) as a way to deal with the selfishness of these agents. Essentially, this definition states that to care for the incentives of these agents, we need not only pay attention about the relationship between  input and  output, but also about the way the algorithm is run. However, it is not clear to date what algorithmic approaches ought to be used for OSP. In this paper, we rather surprisingly show that, for binary allocation problems, OSP is fully captured by a natural combination of two well-known and extensively studied algorithmic techniques: forward and reverse greedy. We call two-way greedy this underdeveloped algorithmic design paradigm.

Our main technical contribution establishes the connection between OSP and two-way greedy. We build upon the recently introduced cycle monotonicity technique for OSP \cite{esa19}. By means of novel structural properties of cycles and queries of OSP mechanisms, we fully characterize these mechanisms in terms of extremal implementations. These are protocols that ask each agent to consistently separate one extreme of their domain at the current history  from the rest. Through the natural connection with the greedy paradigm, we are able to import a host of known approximation bounds to OSP and strengthen the strategic properties of this family of algorithms. Finally, we begin exploring the full power of two-way greedy (and, in turns, OSP) in the context of set systems. 
\end{abstract}

\newpage

\setcounter{page}{1}

\section{Introduction}\label{sec:Introduction}
An established line of work in computer science recognizes the important role played by self interests. If ignored, these self interests can misguide the algorithm or protocol at hand and lead to suboptimal outcomes. Mechanism design has emerged as the framework of reference to deal with this selfishness. Mechanisms are protocols that interact with the selfish agents involved in the computation; the information elicited through this interaction is used to choose a certain outcome (via an algorithm). The goal of a mechanism is that of reconciling the potentially contradictory aims of agents with that of the designer (i.e., optimize a certain objective function). The agents attach a utility (typically defined as quasi-linear function of the \emph{transfers} defined by the mechanism and the agent's \emph{type} -- i.e., cost or valuation -- for the solution) to each outcome and are therefore incentivized to force the output of an outcome that maximizes their utility (rather than maximizing the objective function). The quality of a mechanism is assessed against how well it can approximate the objective function whilst giving the right incentives to the agents.

In this context, one seeks to design \emph{strategyproof (SP) mechanisms} --- these guarantee that agents will not strategize as it will in their best interest to adhere to the rules set by the mechanism --- and aims to understand what is the best possible approximation that can be computed for the setting of interest. For example, it is known how for \emph{utilitarian} problems (roughly speaking, those whose objective function is the sum of all the agents' types) it is possible to simultaneously achieve optimality and strategyproofness, whilst some non-utilitarian objective (such as, min-max) cannot be approximated too well (irrespectively of computational considerations), see, e.g., \cite{book}. These results can be  proved purely from an algorithmic perspective -- that ignores incentives and selfishness -- in that it is known how strategyproofness is equivalent to a certain \emph{monotonicity property} of the algorithm used by the mechanism to compute the outcome. This monotonicity relates the outcomes of two instances, connected by SP constraints, and limits what the algorithm can do on them. For example, if an agent is part of the solution computed on instance $I$ and becomes ``better'' (e.g., faster) in instance $I'$ then the algorithm must select the agent also in the solution returned for instance $I'$, all other things unchanged.

Recent research in mechanism design has highlighted how cognitive limitations of the agents might mean that SP is too weak a desideratum for mechanisms. Even for the simplest setting of one-item auction, there is experimental evidence that people strategize against the sealed-bid implementation of second-price auction, whilst ascending-price auction seems easier to grasp \cite{ausubel2004,kagel87}. The concept of \emph{obvious strategyproofness} (OSP) has been defined in \cite{liosp} to capture this particular form of imperfect rationality, which is shown to be equivalent to the absence of contingent reasoning skills. Intuitively, for an agent it is obvious to understand if a strategy is better than another in that the worst possible outcome for the former is better than the best one for the latter. 

\medskip
\centerline{
	\begin{minipage}[c]{0.8\columnwidth}
		\textit{Can we, similarly to SP, derive bounds on the quality of OSP mechanisms that are oblivious to strategic considerations?}
	\end{minipage}
}
\medskip

There are two obstacles to getting a fully algorithmic approach to OSP mechanisms due to their structure. Whereas SP mechanisms are pairs comprised of an algorithm and transfer (a.k.a., payment) function, in OSP we have a third component -- the so-called \emph{implementation tree} -- which encapsulates the execution details (e.g., sealed bid vs ascending price) of the mechanisms and the obviousness of the strategic constraints. (For OSP, in fact, the implementation details matter and the classical Revelation Principle does not hold \cite{liosp}.) A technique, known as cycle monotonicity (CMON), allows to express the existence of SP payments for an algorithm in terms of the weight of the cycles in a suitably defined graph. Specifically, it is known that it is sufficient to look at cycles of length two for practically all optimization problems of interest \cite{SY05} --- this yields the aforementioned property of monotone algorithms. Recent work \cite{esa19,wine19} extends CMON to OSP and allows to focus only on algorithms and implementation trees. Whilst this has allowed some progress towards settling our main question in the context of single-parameter agents, some unsatisfactory limitations are still present. Firstly, handling two interconnected objects, namely algorithm and implementation tree, simultaneously is hard to work with: e.g., novel ad-hoc techniques (dubbed CMON two-ways in \cite{esa19}) had to be developed to prove lower bounds. Secondly, the CMON extension to OSP is shown to require the study of cycles of any length, thus implying that the ``monotonicity'' of the combination algorithm/implementation tree needs to hold amongst an arbitrary number of instances, as opposed to two as in the case of SP. Thirdly, the mechanisms constructed in \cite{esa19,wine19} only work for small domains (of size up to three) since they rely on the simpler two-instance monotonicity (referred to as monotonicity henceforth). 

\paragraph*{Our contributions.} 
The technical challenge left open by previous work was to relate monotonicity to many-instance monotonicity. In this paper, we solve this challenge by providing a characterization of OSP mechanisms for binary allocation problems (for which the outcome for each agent is either to be selected or not). This enables us to show that the shape of the implementation tree is essentially fixed and answer the question above in the positive. It turns out that the exact algorithmic structure of OSP mechanisms is intimately linked with a (slight generalization of a) well know textbook paradigm:

\medskip
\centerline{
	\begin{minipage}[c]{0.8\columnwidth}
		\em OSP can be achieved if and only if the algorithm is two-way greedy.
	\end{minipage}
}
\medskip

\noindent What does it mean for an algorithm to be two-way greedy? The literature in computer science and approximation algorithms 
has extensively explored what we call \emph{forward greedy}. These are algorithms that use a (possibly adaptive) (in-)priority function and incrementally build up a solution by adding therein the agent with the highest priority, if this preserves feasibility. It is known that if the priority rule is monotone in each agent's type then this leads to a SP direct-revelation mechanism (see, e.g., \cite{LOS}). What we show here is that the strategyproofness guarantee is actually much stronger and can deal with imperfect rationality. This is achieved with a simple implementation of forward greedy that sweeps through each agent's domain from the best possible type to the worst. 
Another relevant approach known in the literature is Deferred Acceptance auctions (DAAs) or reverse greedy algorithms \cite{DGR17,GMR17}. These use a (possibly adaptive) (out-)priority function and build a feasible solution by incrementally throwing out the agents whose type is not good enough with respect the current priority (i.e., whose cost (valuation) is higher (lower) than the out-priority) until a feasible solution is found. It is already known that DAAs are OSP \cite{MS20} but not the extent to which focusing on them would be detrimental to finding out the real limitations of OSP mechanisms. Two-way greedy algorithms \emph{combine in- and out-priorities}; each agent faces either a greedy in- or  out-priority; in the former case, they are included in the solution if feasibility is preserved while in the latter they are excluded from it if the current solution is not yet feasible. The direction faced can depend on which agents have been included in or thrown out from the eventual solution at that point of the execution; in this sense, these are particular adaptive priority algorithms. For a formal definition, please see Section \ref{sec:2way} and Algorithm \ref{alg:2greedy} below.

Two-way greedy algorithms stem from our characterization of OSP mechanisms 
in terms of ``extremal implementation trees''; roughly speaking, in these mechanisms we always query each agents about \emph{(the same) extreme} of their domain at the current history.
To prove this characterization, we first give a couple of structural properties of OSP mechanisms. We specifically show (i) when a query can be made to guarantee OSP; and, (ii) how a mechanism that is monotone but not many-instance monotone looks like. We use the former property to show that, given an OSP mechanism, we can modify the structure of its implementation tree to make it extreme whilst guaranteeing that the many-instance monotonicity is preserved (i.e., the structure (ii) is not possible). We also show that extremal mechanisms are monotone and that structure (ii) can never arise, thus proving the sufficient condition of our characterization.\footnote{Interestingly, property (ii) proves why for domains of size up to three, it is possible to freely interleave queries as in \cite{esa19,wine19} since it shows that four types are necessary for the two-instance monotonicity to become insufficient (and only necessary).}  

One caveat about these extremal mechanisms and two-way greedy is necessary. This has to do with a technical exception to the rule of never interleaving top queries (asking for the maximum of the current domain) with bottom queries (asking for the minimum) to an agent. An OSP mechanism can in fact interleave those when, at the current history, an agent becomes \emph{revealable}, that is, the threshold separating winning bids from losing ones, becomes known. In other words, this is a point in which the outcome for this agent (but not necessarily the entire solution) is determined for all but one of her types. OSP mechanisms can at this point use any query ordering to find out what the type of the agent is; this does not affect the incentives of the agents. Accordingly, in a two-way greedy algorithm an agent can face changes of priority direction (e.g., from in- to out-priority) in these circumstances.

To the best of our knowledge this is the first known case of a relationship between strategic properties and an algorithmic paradigm, as opposed to a property about the solution output by the algorithm. Two possible interpretations of this connection can be given. On a conceptual level, the fact that the Revelation Principle does not hold true for OSP means that we care about the implementation details and therefore the right algorithmic nature has to be paradigmatic rather than being only about the final output. On a technical level, given that monotonicity (i.e., two-cycles) is not sufficient for OSP, the need to study many-instance monotonicity (i.e., any cycle) requires to go beyond an output property and look for the way in which the algorithm computes \textit{any} solution.

It is worth noticing that our OSP characterization is related to the Personal Clock Auctions (PCAs) in \cite{liosp}. Li proves that for binary allocation problems \emph{and continuous  domains}, each agent faces either an ascending-price auction (where there is an increasing transfer going rate to be included in the solution) or a descending-price auction (where there is a decreasing transfer going rate to be excluded from the solution). Our result complements Li's along two dimensions. Firstly, it generalizes the characterization to the case of finite domains. This is arguably more interesting for OSP, as it is notoriously harder to understand how to execute extensive-form games in the continuous case. Moreover, our proof does not rely on the existence of a unique threshold (in the case of discrete domains, there are two values that satisfy this definition, i.e., extreme winning and losing reports do not ``meet'' in the limit). Secondly, our approach and terminology are closer to computer science and algorithms and give a specific recipe to reason about design and analysis of these mechanisms. For example, it is not clear how to precisely define the transfer going rates in PCAs and how this relates to the quality of their algorithmic outcome.  


We give a host of bounds on the approximation guarantee of OSP mechanisms by relying on our characterization and the known approximation guarantees of forward greedy algorithms, cf. Table \ref{tbl:upperbounds} below. The strategic equivalence of forward and reverse greedy is one of the most far reaching consequences of our results, given {(i)} the rich literature on the approximation of forward greedy, and {(ii)} the misconception about the apparent weaknesses of accepting, rather than rejecting, auctions \cite{MS20} (see Section \ref{sec:2way} and Appendix \ref{sec:cas} for details). We expect our work to spawn further research about OSP, having fully extracted the algorithmic nature of these mechanisms. The power and limitations of OSP can now be fully explored, in the context of binary allocation problems. We present some initial bounds on the quality of these algorithms/mechanisms (see Sections \ref{sec:2way-1way} and \ref{sec:2way-apx}). Notably, we close the gap for the approximation guarantee of OSP mechanisms for the so-called knapsack auctions, studied in \cite{DGR17}. We show that the logarithmic upper bound provided by the authors is basically tight not just for reverse greedy (as shown by them) but for the whole class of two-way greedy/OSP. 

Since our main objective is that of establishing the power of OSP, in terms of algorithmic tools and their approximation, we do not primarily focus on the computational complexity of these mechanisms. Consequently, our lower bounds are unconditional. 
We discuss this aspect and more opportunities for further research in the conclusions (see Section \ref{sec:conclusions}). A discussion on more related work can be found in Appendix \ref{apx:rel}.

\section{Preliminaries and Notation}
For the design of a mechanism, we need to define a set $N$ of $n$ \emph{selfish agents} and a set of feasible \emph{outcomes} $\out$.
Each agent $i$ has a \emph{type} $t_i \in D_i$, where $D_i$ is the \emph{domain} of $i$.
The type $t_i$ is assumed to be \emph{private knowledge} of agent $i$.
We let $t_i(X) \in \mathbb{R}$ denote the \emph{cost}
of agent $i$ with type $t_i$ for the outcome $X \in \out$.
When costs are negative, it means that the agent has a {profit} from the solution, called \emph{valuation}.

A \emph{mechanism} has to select an outcome $X \in \out$.
For this reason, the mechanism interacts with agents.
Specifically, agent $i$ takes \emph{actions} (e.g., saying yes/no)
that may depend on her presumed type $b_i \in D_i$
(e.g., saying yes could ``signal'' that the presumed type has some properties that $b_i$ enjoys).
To stress this we say that agent $i$ takes \emph{actions compatible with (or according to) $b_i$}.
Note that the presumed type $b_i$ can be different from the real type $t_i$.

For a mechanism $\M$, we let $\M(\b)$ denote the outcome returned by the mechanism
when agents take actions according to their presumed types $\b = (b_1, \ldots, b_n)$  (i.e., each agent $i$ takes actions compatible with the corresponding $b_i$).
This outcome is given by a pair $(f,p)$,
where $f = f(\b)=(f_1(\b),\ldots,f_n(\b))$ (termed \emph{social choice function} or, simply, algorithm) maps the actions taken by the agents according to $\b$ to a feasible solution in $\out$,
and $p=p(\b)=(p_1(\b),\ldots,p_n(\b)) \in \mathbb{R}^n$ maps the actions taken by the agents according to $\b$ to \emph{payments}. Note that payments need not be positive. 

Each selfish agent $i$ is equipped with a \emph{quasi-linear utility function} $u_i \colon D_i \times \out \rightarrow \mathbb{R}$:
for $t_i \in D_i$ and for an outcome $X \in \out$ returned by a mechanism $\M$, $u_i(t_i, X)$ is the utility that agent $i$ has for the implementation of outcome $X$ when her type is $t_i$, i.e., 
$
u_i(t_i, \M(b_i, \bi)) = p_i(b_i, \bi) - t_i(f(b_i, \bi)).
$

In this work we will focus on \emph{single-parameter} settings, that is, the case in which
the private information of each bidder $i$ is a single real number $t_i$ and $t_i(X)$ can be expressed as $t_i w_i(X)$ for some publicly known function $w_i$.
To simplify the notation, we will write $t_i f_i(\b)$ when we want to express the cost of a single-parameter agent $i$ of type $t_i$ for the output of social choice function $f$ on input the actions corresponding to a bid vector $\b$. In particular, we will consider \emph{binary allocation problems}, where $f_i(\b) \in \{0, 1\}$, i.e., each agent either belongs to the returned solution ($f_i(\b)=1$) or not ($f_i(\b)=0$). A class of binary allocation problems of interest are \emph{set systems} $(E, {\mathcal{F}})$, where $E$ is a set of elements and $\mathcal{F} \subseteq 2^E$ is a family of feasible subsets of $E$. Each element $i \in E$ is controlled by a selfish agent, that is, the cost for including $i$ is known only to agent $i$ and is equal to some value $t_i$. The social choice function $f$ must choose a feasible subset of elements $X$ in $\mathcal{F}$ that minimizes $\sum_{i=1}^n b_i(X)$. When the $t_i$'s are non-negative (non-positive, respectively) then this objective is called social cost minimization (social welfare maximization, respectively).

\subsection{Extensive-form Mechanisms and Obvious Strategyproofness}
We now introduce the concept of implementation tree and we formally define (deterministic) obviously strategy-proof mechanisms. Our definition is built on \cite{mackenzie} rather than the original definition in~\cite{liosp}. Specifically, our notion of implementation tree is equivalent to the concept of round-table mechanisms in \cite{mackenzie}, and our definition of OSP is equivalent to the concept of SP-implementation through a round table mechanism, that is proved to be equivalent to the original definition of OSP.

Let us first formally model how a mechanism works. An \emph{extensive-form mechanism} $\M$ is a triple $(f,p,\T)$ where, as above, the pair $(f,p)$ determines the outcome of the mechanism, and $\T$ is a 
tree, called \emph{implementation tree}, such that:
\begin{itemize}[leftmargin=0.45cm, noitemsep]
	\item Every leaf $\ell$ of the tree is labeled with a possible outcome of the mechanism $(X(\ell), p(\ell))$, where $X(\ell) \in \out$ and $p(\ell) \in \mathbb{R}$;
	\item Each node $u$ in the implementation tree $\T$ defines the following:
    \begin{itemize}[leftmargin=0.6cm, noitemsep]
        \item An agent $i=i(u)$ to whom the mechanism makes some query. Each possible answer to this query leads to a different child of $u$.
        \item A subdomain $D^{(u)}=(D_i^{(u)}, D_{-i}^{(u)})$ containing all types that are \emph{compatible} with $u$, i.e., with all the answers to the queries from the root down to node $u$.
        Specifically, the query at node $u$ defines a partition of the current domain of $i$, $D_i^{(u)}$ into $k\geq 2$ subdomains, one for each of the $k$ children of node $u$. Thus, the domain of each of these children will have as the domain of $i$, the subdomain of $D_i^{(u)}$ corresponding to a different answer of $i$ at $u$, and an unchanged domain for the other agents.
\end{itemize}
\end{itemize}

Observe that, according to the definition above, for every profile $\b$ there is only one leaf $\ell = \ell(\b)$
such that $\b$ belongs to $D^{(\ell)}$.
Similarly, to each leaf $\ell$ there is at least a profile $\b$ that belongs to $D^{(\ell)}$.
For this reason, we say that $\M(\b) = (X(\ell), p(\ell))$.

Two profiles $\b$, $\b'$ are said to \emph{diverge} at a node $u$ of $\T$ if this node has two children $v, v'$ such that $\b \in D^{(v)}$, whereas $\b' \in D^{(v')}$. For every such node $u$, 
we say that 
$i(u)$ is the \emph{divergent agent} at $u$.

We are now ready to define obvious strategyproofness.
An extensive-form mechanism $\M$ is \emph{obviously strategy-proof (OSP)} if for every agent $i$ with real type $t_i$,
for every vertex $u$ such that $i = i(u)$, for every $\bi, \bi'$ (with $\bi'$ not necessarily different from $\bi$),
and for every $b_i \in D_i$, with $b_i \neq t_i$,
such that $(t_i, \bi)$ and $(b_i, \bi')$ are compatible with $u$, but diverge at $u$,
it holds that $$u_i(t_i, \M(t_i, \bi)) \geq u_i(t_i,\M(b_i, \bi')).$$
Roughly speaking, an obviously strategy-proof mechanism requires that, at each time step
agent $i$ is asked to take a decision that depends on her type, the worst utility that
she can get if she behaves according to her true type
is at least the best utility she can get by behaving differently.
{We stress that our definition does not restrict the alternative behavior to be consistent with a fixed type. Indeed, as noted above, each leaf of the tree rooted in $u$, denoted $\T_u$, corresponds to a profile $\b = (b_i, \bi')$ compatible with $u$: then, our definition implies that the utility of $i$ in the leaves where she plays truthfully is at least as much as the utility in every other leaf of $\T_u$.}


\subsection{Cycle-monotonicity Characterizes OSP Mechanisms}
We next describe the main tools in \cite{esa19} showing that  OSP can be characterized by the absence of negative-weight cycles in a suitable weighted graph over the possible strategy profiles. For ease of exposition, we will focus on non-negative costs but the results hold no matter the sign. We consider a mechanism $\M$ with implementation tree $\T$ for a social choice function $f$, and define the following concepts: 
\begin{itemize}[leftmargin=0.45cm, noitemsep]
	\item \textbf{Separating Node:} A node $u$ in the implementation tree $\T$ is  $(\a,\b)$-separating for agent $i=i(u)$ if $\a$ and $\b$ are compatible with $u$ (that is, $\a,\b\in D^{(u)}$), and the two types $a_i$ and  $b_i$ belong to two different subdomains of the children of $u$ (thus implying $a_i\neq b_i$).
	\item \textbf{OSP-graph:} For every agent $i$, we define a 
	 directed weighted graph $\ver$ having a node for each profile in $D=\times_i D_i$. The graph 
	contains 
	edge $(\a, \b)$ if and only if $\T$ has some node $u$ which is $(\a,\b)$-separating for $i=i(u)$, and the weight of this edge is $w(\a, \b)= a_i(f_i(\b)- f_i(\a))$. Throughout the paper, we will denote with $\a \rightarrow \b$ an edge $(\a,\b) \in \ver$, and with $\a \rightsquigarrow \b$ a path among these two profiles in $\ver$.
	\item \textbf{OSP Cycle Monotonicity (OSP CMON):} We say that the OSP cycle monotonicity (OSP CMON)  holds  if, for all $i$, the graph $\ver$ does not contain negative-weight cycles. Moreover, we say that the OSP two-cycle monotonicity (OSP 2CMON) holds if the same is true when considering cycles of length two only, i.e., cycles with only two edges.
\end{itemize}

\begin{theorem}[\cite{esa19}]\label{thm:cmon}
	A mechanism with implementation tree $\T$ for a social function $f$ is OSP on finite domains if and only if  OSP CMON holds.
\end{theorem}
Given the theorem above, we henceforth assume that the agents have finite domains.

\section{A  Characterization of OSP Mechanisms}
In this section we present our characterization of OSP mechanisms for binary allocation problems. 
\subsection{Setting the Scene}
It is useful to focus on the possible weights of single edges as well as on two-cycles, in the case of binary allocation problems. 
\begin{observation}[Basic Properties of the OSP-graph] \label{obs:OSP-graph}
	The weight of an edge is non-zero only for a: 
	\begin{itemize}[leftmargin=0.45cm, noitemsep]
		\item \textbf{positive-weight edge}, i.e., $f_i(\a)=0$ and $f_i(\b)=1$, in which case the weight is $w(\a,\b) = +a_i$;
		\item \textbf{negative-weight edge}, i.e., $f_i(\a)=1$ and $f_i(\b)=0$, in which case the weight is $w(\a,\b) = -a_i$. 
	\end{itemize}
	If OSP 2CMON holds, then (i) for every positive-weight edge as above, we have $b_i < a_i$ and, by symmetry, (ii) for every negative-weight edge   as above, we have $a_i < b_i$.  Note that these inequalities are strict since $a_i\neq b_i$ for every edge as above.
\end{observation}

The following notation will also be used throughout our proofs. 
\begin{definition}[\always{0}, \always{1}, \neither]\label{def:typeclas}
	For $\sel  \in \{0,1\}$, indicating whether $i$ is selected, and for a generic type $t_i \in D_i^{(u)}$, for a node $u$ of the implementation tree, we define the following:
	\begin{itemize}[leftmargin=0.45cm, noitemsep]
		\item Type $t_i$ is $\always{\sel}$  if  $f_i(\vect t)=\sel$ 
		for all $\ti \in D_{-i}^{(u)}$; 
		\item Type $t_i$ is $\sometime{\sel}$ if $f_i(\vect t)=\sel$ for some $\ti \in D_{-i}^{(u)}$.
	\end{itemize} 
	Moreover, type $t_i$ is $\neither$ if it is neither \always{0} nor \always{1}, that is, it is both \sometime{0} and \sometime{1}.
\end{definition}

We next give three structural properties of OSP mechanisms that will be useful. 

\paragraph{Structure of the Implementation Tree.} We first observe that, without loss of generality, we can always assume that each node $u$ in the implementation tree has two children, and thus the mechanism partitions $D_i^{(u)}$ into two subdomains $L^{(u)}$ and $R^{(u)}$.

\begin{observation}\label{obs:binary}
	For any OSP mechanism $\M=(f, p, \T)$ where $\T$ is not a binary tree, there is an OSP mechanism $\M'=(f, p, \T')$ where $\T'$ is a binary tree.
\end{observation}
We refer the interested reader to Appendix~\ref{apx:binary} for a proof of Observation~\ref{obs:binary}.

We stress  that, in general, the two parts $L^{(u)}$ and $R^{(u)}$ can be \emph{any} partition of the subdomain $D_i^{(u)}$, and they are not necessarily ordered. For example, if $D_i^{(u)} = \{1,2,5,6,7\}$, a query ``is your type even?'' results in $L^{(u)} = \{1,5,7\}$ and $R^{(u)}=\{2,6\}$,  with these two subdomains being ``incomparable''. However, we will see that in an OSP mechanism these sets $L^{(u)}$ and $R^{(u)}$ share a special structure.

\paragraph{Structure of Admissible Queries.} We begin with a simple observation about the structure of the parts implied by a simple application of OSP 2CMON.

\begin{observation}\label{obs:2CMON}
	If OSP 2CMON holds and  at every node $u$ where the subdomain is separated into two parts $L$ and $R$ the following holds. Every $\sometime{0}$ element $l$ in one side (say $L$) implies that  all elements in the other side that are bigger ($r\in R$ with $r>l$) must be \always{0}. Similarly, every $\sometime{1}$ element $l$ in one side (say $L$) implies that  all elements in the other side that are smaller ($r\in R$ with $r<l$) must be \always{1}.
\end{observation}

Next lemma characterizes the admissible queries in OSP mechanisms.

\begin{lemma}[Admissible Queries]\label{le:local}
	Let $\M$ be an OSP mechanism with implementation tree $\T$ and let $u$ be a node of $\T$ where the query separates the current subdomain into two parts $L$ and $R$. Then one of the following conditions must hold: 
	\begin{enumerate}
		\item \label{cond:homogeneous} At least one of the two parts, some $P \in \{L,R\}$, is homogeneous meaning that
\begin{align*}
		P = \{\underbrace{p_1<p_2<\cdots< p_{|P|}}_{\always{1}}\} && \text{or} && P = \{\underbrace{p_1<p_2<\cdots< p_{|P|}}_{\always{0}}\}
		\end{align*}
		\item \label{cond:revealable} Agent $i$ is  revealable at node $u$, meaning that the subdomain $D_i^{(u)}$ has the following structure:
		\begin{align}\label{eq:revealable}
		D_i^{(u)} = \{\underbrace{d_1<d_2<\cdots<d_a}_{\always{1}} < d_* < \underbrace{d_1' < d_2' < \cdots d_b'}_{\always{0}}\}
		\end{align}
		where each subset of \always{1} and \always{0} types may be empty, and $d_*$ may or may not be \always{1} or \always{0}. Moreover, the two parts that are separated must have the following structure:
		\begin{align}\label{eq:admi-query-L}
		L = \{\underbrace{l_1  < l_2<\cdots < l_a}_{\always{1}} < l_*< \underbrace{l_1'<l_2'\cdots < l_b'}_{\always{0}} \} \ ,\\
		\label{eq:admi-query-R}
		R = \{\underbrace{r_1  <r_2< \cdots < r_c}_{\always{1}} < r_* < \underbrace{r_1'<r_2'\cdots < r_d'}_{\always{0}}\} \ ,
		\end{align}
		where at most one between $l_*$ and $r_*$ is $\neither$ (neither \always{0} nor \always{1}).
	\end{enumerate}
\end{lemma}

\begin{proof}
	We show that, if neither $L$ nor $R$ is homogeneous, then they must satisfy \eqref{eq:admi-query-L} and \eqref{eq:admi-query-R}. Using OSP 2CMON (Observation~\ref{obs:2CMON}), we can show that $L$ must have the following structure (the symmetric argument applies to $R$):
	\begin{align*}
	L = \{\underbrace{l_1  < l_2<\cdots < l_a}_{\always{1}} < \underbrace{l_1'  < l_2'<\cdots < l_a'}_{\sometime{1}}   < \underbrace{l_1''  < l_2''<\cdots < l_a''}_{\sometime{0}} < \underbrace{l_1'''<l_2'''\cdots < l_b'''}_{\always{0}} \} \ ,
	\end{align*}
	where each of these four subsets may be empty. Since $R$ is not homogeneous, there are two profiles $\vect r^0,\vect r^1$ with $f_i(\vect r^1)=1$ and $f_i(\vect r^0)=0$. We next show that $L$ must contain at most one \neither\ type $l_*$, hence \eqref{eq:admi-query-L} holds. 
	By contradiction, suppose there are two \neither\ types  $l_*$ and $l'_*$ with $l_*<l_*'$. By definition, these types are neither \always{0} nor \always{1}. Therefore, there are two profiles satisfying  $f_i(\vect l_*)=0$ and $f_i(\vect l_*')=1$. Since all edges between $L$ and $R$ are present in the graph $\ver$, the latter contains the cycle $\vect r^1 \rightarrow \vect l_*' \rightarrow \vect r^0 \rightarrow \vect l_*\rightarrow \vect r^1$ whose weight is $-l_*+l_*'<0$. This contradicts OSP CMON and thus the hypothesis that the mechanism was OSP.
	
	We next observe that OSP 2CMON implies that there cannot be
	$l_*\in L$ and $r_*\in R$ such that both of them are \neither. For $l_* < r_*$,  we consider  two profiles such that $f_i(\vect l_*)=0$ and $f_i(\vect r_*)=1$ (these profiles exist since these types are neither \always{0} nor \always{1}). The cycle between them has weight $l_* - r_*<0$, contradicting OSP 2CMON. For $l_* > r_*$, we use two different profiles for which $f_i(\vect l_*')=1$ and $f_i(\vect r_*')=0$ (which again exist by the hypothesis on $r_*$ and $l_*$).
	
	Finally, we observe that \eqref{eq:admi-query-L}-\eqref{eq:admi-query-R} and OSP 2CMON (Observation~\ref{obs:2CMON}) imply \eqref{eq:revealable}. 
\end{proof}

\paragraph{Structure of Negative-Weight Cycles.} 
A crucial step is to provide  a simple \emph{local-to-global} characterization of OSP mechanisms which (essentially) involves only four type profiles.
The following theorem states that, if there is a negative-weight cycle with \emph{more than two} edges but no length-two negative cycle  (OSP CMON is violated but OSP 2CMON holds), then there exists a cycle with a rather \emph{special structure}, and this special structure is fully specified by only \emph{four} profiles (the cycle itself may involve several profiles though). 
\begin{theorem}[Four-Profile Characterization]\label{thm:anatomy}
	Let $\M$ be a mechanism with implementation tree $\T$ and social choice function $f$ that is  OSP 2CMON but not OSP CMON. 
	Then, every negative-weight cycle $C$ in some OSP graph $\ver$ is of the following form:
	\begin{equation}\label{eq:cycle-four}
	C = \bb 2 \rightarrow \bb 1 \rightsquigarrow  \bb 3 \rightarrow \bb 4 \rightsquigarrow \bb 2 
	\end{equation}
	where these four profiles satisfy  $\bbb 1 < \bbb 2
	< \bbb 3 < \bbb 4$,
	  $f_i(\bb 1) = f_i(\bb 3) = 1$, and  $f_i(\bb 2) = f_i(\bb 4) = 0$. 
Moreover, there is no edge between $\bb 2$ and $\bb 3$ in $\ver$.
\end{theorem}
  \begin{proof}
  	The key observation is that every negative-weight cycle $C$ must contain both a positive-weight edge $(\c^{(pos1)},\c^{(pos2)})$ and a negative-weight edge $(\c^{(neg1)},\c^{(neg2)})$ that satisfy $\ccc{pos1} < \ccc{neg1}$. This is because the number of positive-weight edges equals the number of negative-weight edges in any cycle, and the weight of these two edges is $\ccc{pos1}$ and $-\ccc{neg1}$, respectively (Observation~\ref{obs:OSP-graph}).  
  	Hence, cycle $C$ is of the form
  	\begin{align}\label{eq:cycle-four-proof}
  	\c^{(pos1)} \rightarrow \c^{(pos2)} \rightsquigarrow  \c^{(neg1)} \rightarrow \c^{(neg2)} \rightsquigarrow \c^{(pos1)}\ ,
  	\end{align}
  	where Observation~\ref{obs:OSP-graph} implies that $f_i(\c^{(pos1)}) =  0 = f_i(\c^{(neg2)})$ and $f_i(\c^{(pos2)}) =  1 = f_i(\c^{(neg1)})$, and also
  	\begin{align*}
  	\ccc {pos2} < \ccc {pos1} < \ccc {neg1} < \ccc {neg2} \ ,
  	\end{align*}
  	where the first and the last inequalities follow from  Observation~\ref{obs:OSP-graph} together with OSP 2CMON. This implies the first part of the theorem by matching the profiles in \eqref{eq:cycle-four} with those in  \eqref{eq:cycle-four-proof}.
  	
  	We finally observe that OSP 2CMON implies that there is no edge between $\b^{(2)} = \c^{(pos1)}$ and $\b^{(3)}=\c^{(neg1)}$ because otherwise the corresponding 2-cycle between these profiles would have negative weight: Indeed, 
  	$w(\c^{(pos1)}\rightarrow \c^{(neg1)})+  w(\c^{(neg1)} \rightarrow \c^{(pos1)}) = \ccc{pos1} - \ccc{neg1}<0$ since $f_i(\c^{(pos1)}) =  0$ and $f_i(\c^{(neg1)})=1$. This completes the proof.
\end{proof}
The characterization of OSP CMON above will enable us to provide a simple local transformation of the queries of an OSP mechanisms where we use only top or bottom queries.
 
\begin{definition}[Top and Bottom Queries]
Let $i=i(u)$ for a node $u$ of the implementation tree. If the query at $u$ partitions $D_i^{(u)}$ into $\{\min D_i^{(u)}\}$ and $D_i^{(u)} \setminus \{\min D_i^{(u)}\}$ then we call the query at $u$ a bottom query. A top query at $u$, instead, separates the maximum of $D_i^{(u)}$ from the rest.
\end{definition}	 

 \subsection{OSP is Equivalent to Weak Interleaving}
In this section, we show that without loss of generality, we can focus on OSP mechanisms where each agent is asked only top queries or only bottom queries, except when her type becomes revealable (Condition \ref{cond:revealable} in Lemma~\ref{le:local}). In that sense, these mechanisms interleave top and bottom queries for an agent only in a ``weak'' form.
Specifically, let us begin by providing the following definition.
 \begin{definition}[Extremal, No Interleaving, Weak Interleaving]\label{def:extr-and-interleaving}
 	A mechanism is extremal if every query is a bottom query or a top query (both types of queries may be used for the same agent). 
 	
 	An extremal mechanism makes no interleaving if each agent is consistently asked only at each history top queries or only bottom queries (some agents may be asked top queries only, and other agents bottom queries only). 
 	
 	A weak interleaving mechanism satisfies the condition that, if top queries and bottom queries are interleaved for some agent $i$ at some node $u$ in the implementation tree, then agent $i$ is revealable at $u$ in the sense of Condition~\ref{cond:revealable} in Lemma~\ref{le:local}.
 \end{definition}
 Then, in this section we will prove the following theorem:
 \begin{theorem}
  For each binary outcome problem, an OSP mechanism exists if and only if an extremal mechanism with weak interleaving exists.
 \end{theorem}

 \subsubsection{The Necessary Condition}
 We prove that the following transformations are always possible:
 \[
 OSP~mechanism \Rightarrow OSP~extremal~mechanism \Rightarrow OSP~weak~ interleaving  ~mechanism.
 \] 
  Intuitively, given the structure of admissible queries in Lemma~\ref{le:local}, we show below that we can locally replace every query with a homogeneous part (Condition \ref{cond:homogeneous}) by an ``homogeneous'' sequence of only top queries or only bottom queries. Moreover, these queries can also be used when the agent becomes revealable (Condition~\ref{cond:revealable}). To this aim, we use the four-profile characterization of negative-weight cycles in Theorem~\ref{thm:anatomy}.

 \begin{theorem}\label{th:ops-extremal}
 	Any OSP mechanism $\M=(f, p, \T)$ can be transformed into an equivalent extremal OSP mechanism $\M'=(f, p, \T')$.
 \end{theorem}
 \begin{proof}
Since the mechanism $\M$ is OSP, then OSP CMON must hold (Theorem~\ref{thm:cmon}). Let $u$ be a node of $\T$ where $\M$ is not extreme, and let $i=i(u)$ be the corresponding divergent agent at $u$. 

We locally modify $\T$ in order to make a suitable sequence of bottom queries and top queries about a certain subset $Q$ of types, before we make any further query in the two subtrees of $u$. This local modification does not affect OSP CMON of agents different from $i$ and it preserves OSP 2CMON for all agents (see Appendix~\ref{apx:local_tr} for details on this step).
The proof uses the following main steps and key observations:
\begin{enumerate}[leftmargin=0.45cm, noitemsep]
	\item If OSP CMON is no longer true for $\T'$, then there exists a negative-weight cycle $C'$ which was not present in $\ver$ and therefore must use some added edges $\ab 1 \rightarrow \ab 2$ that were not present in $\ver$ and that have been added to $\verp$ because of the new queries for types in $Q$;
	\item The negative-weight cycle $C'$ must be of the form specified by Theorem~\ref{thm:anatomy}; in particular, the edge $\bb 2 \rightarrow \bb 3$ does not exist in $\verp$ (which helps to determine properties of the four profiles characterizing $C'$);
	\item We use the following \emph{bypass argument} to conclude that in the original graph there is a negative-weight cycle, thus contradicting OSP CMON of the original mechanism. Specifically, for every added edge $\ab 1 \rightarrow \ab 2$, the original graph $\ver$ contains a bypass path $\ab 1 \rightarrow \bb {bp} \rightarrow \ab 2$ such that $w(\ab 1 \rightarrow \bb{bp} \rightarrow \ab 2) \leq  w(\ab 1 \rightarrow \ab 2)$. By replacing every added edge in $C'$ with the corresponding bypass path, we get a cycle  $C$ with negative weight, $w(C)\leq w(C')<0$.
\end{enumerate}
In the remainder of the proof, we distinguish the two possible cases on the structure of $L$ and $R$ at $u$ in $\T$ using the admissible queries characterization of Lemma~\ref{le:local}:

\begin{description}[noitemsep]
\item[Homogeneous case~(Condition~\ref{cond:homogeneous} in Lemma~\ref{le:local}).] Suppose $R$ is the homogeneous part. We distinguish the two sub-cases.
   \begin{itemize}[leftmargin=0.45cm, noitemsep]
    \item \textbf{All $r \in R$ are $\always{0}$.} 
   Let $r_{\min}$ be the minimum type in $R$ and partition $L$ into $L^{(<r_{\min})}:=\{l \in L : l < r_{\min}\}$ and $L^{(\geq r_{\min})}=\{l \in L : l \geq  r_{\min}\}$. By  OSP 2CMON (Observation~\ref{obs:2CMON}),  all types in $Q=R \cup L^{(\geq r_{\min})}$ are $\always{0}$. 
    We perform a sequence of top queries on $Q = R \cup L^{(\geq r_{\min})}$ from the largest to the smallest type. 
   
   In this case we have $\bbb 2 \in L$ since there is no edge between $\bb{2}$ and $\bb{3}$ and $f_i(\bb 3)=1$ implies $\bbb{3}\in L^{(<r_{\min})}$. 
   The bypassing argument  is as follows: 
   \[\bb{bp} = \begin{cases}
   \bb 2 & \text{if } \abb 1, \abb2 \in R \\
   \vect r_{\min}  & \text{if } \abb 1, \abb2 \in L
   \end{cases}
   \]
   where $\vect r_{\min}=(r_{\min}, \vect{r}_{-i})$, for an $\vect{r}_{-i} \in D_{-i}^{(u)}$. Note that the bypass path $\ab 1 \rightarrow \bb{bp} \rightarrow \ab 2$ indeed exists in all cases. 
    Table~\ref{tbl:tinL} shows that the weight of the bypass path is at most the weight of the corresponding edge in all possible cases (note that at least one between $\abb 1$ and $\abb{2}$ must be in  $Q = R \cup L^{(\geq r_{\min})}$ in order for $\ab1 \rightarrow \ab2$ to be a newly added edge). 
    
    \begin{table}[htb]
    	\centering
    	\begin{tabular}{c|c|c|c|c|c}
    		$\abb 1$ & $\abb 2$ & $f_i(\ab 1)$ & $f_i(\ab 2)$ & $w(\ab 1 \rightarrow \ab 2)$ & $w(\ab 1 \rightarrow \bb{bp}  \rightarrow \ab 2)$ \\ 
    		\hline $R$ & $R$ & 0 & 0 & 0 & 0\\
    		\hline $L^{(\geq r_{\min})}$ & $L^{(< r_{\min})}$ &   0   & 1   & $+\abb 1$ &  $r_{\min}$ ~($r_{\min} \leq  \abb 1$)\\
    		\hline $L^{(< r_{\min})}$ & $L^{(\geq  r_{\min})}$ &  1  & 0  & $-\abb 1$ & $-\abb 1$\\
    	\end{tabular}
    	\caption{Bypassing path argument for the homogeneous case with all $r \in R$ that are $\always{0}$.}
    	\label{tbl:tinL}
    \end{table}

    \item \textbf{All $r \in R$ are $\always{1}$.} 
    This can be proved similarly to the previous subcase, by doing  bottom queries on the set $Q= R \cup L^{(< r_{\min})}$, and by observing that OSP 2CMON implies that all types in $Q= R \cup L^{(< r_{\min})}$ are \always{1}.

    \end{itemize}

    \item[Revealable case~(Condition~\ref{cond:revealable} in Lemma~\ref{le:local}).]
    Suppose w.l.o.g. that $L$ contains only \always{1} and \always{0} types (otherwise, we consider $R$ in  place of $L$). We consider $Q=L=L^{(0)} \cup L^{(1)}$,  where $L^{(0)}$ are the $\always{0}$ types of $L$, and $L^{(1)}$ are the $\always{1}$ types. In this case, we do top queries on  $L^{(0)}$ and bottom queries on $L^{(1)}$. The bypassing argument  is as follows: 
    \[\bb{bp} = \begin{cases}
    	\bb 2 & \text{if } \abb{2} \in \LL{0} \\
    	\bb 3 & \text{if } \abb{2} \in \LL{1}
    \end{cases}
    \]
    so that $f_i(\bb{bp}) = f_i(\ab 2)$ and thus $w(\ab{1} \rightarrow \bb{bp} \rightarrow \ab{2})= w(\ab{1} \rightarrow \bb{bp} ) = w(\ab{1} \rightarrow \ab{2} )$ as desired. Note that all edges $\ab{1} \rightarrow \bb{bp} \rightarrow \ab{2}$ exist in the original graph $\ver$ since $\abb{1},\abb{2}\in Q=L$ and $\bbb{2},\bbb{3}\in R$, for otherwise $\bbb{2},\bbb{3} \in L$ and thus edge $\bb 2 \rightarrow \bb 3$ would exist in $\verp$, contradicting the conditions in Theorem~\ref{thm:anatomy}. \qedhere
\end{description}
  \end{proof}

We are now ready to show that one can think of weak interleaving OSP mechanisms without loss of generality.

\begin{theorem}
	Let $\M$ be an OSP  extremal mechanism  with implementation tree $\T$. For any node $u \in \T$, if agent $i=i(u)$  is not revealable at node $u$, then $\M$ has no interleaving for agent $i$ at $u$.
\end{theorem}
\begin{proof}
	By contradiction, assume that $\M$ interleaves the queries of agent $i$ at $u$, and that $i$ is not revealable at $u$. Since there is interleaving then there is a $v \in \T$ in the path from the root of $\T$ to $u$ such that 
	the mechanism makes a bottom query about 
	the minimum $\underline{x}$ of her current domain at $v$, and a top query about her current maximum $\overline x$ of her current domain at $u$, respectively. (The proof is the same when top and bottom are swapped.)

	The proof uses the same argument in the proof of Lemma~\ref{le:local}.
	Using OSP 2CMON, we can show that the bottom query type $\underline{x}$ must be $\always{1}$ at $v$, and the top query type $\overline{x}$ must be $\always{0}$ at $u$.  Since $i$ is not revealable at $u$,  subdomain $D_i^{(u)}$ contains \emph{two} distinct types $t$ and $t'$ that are neither $\always{0}$ nor $\always{1}$ (thus, they are both different from $\underline{x}$ and $\overline{x}$). Assume without loss of generality that $t <  t'$, and observe that there are profiles satisfying
	\begin{align} \label{eq:not-revealable-four-profiles}
	f_i(\underline{\x})=1 && f_i(\vect t')=1 &&  f_i(\overline \x)=0 && f_i(\vect t)=0, 
	\end{align}
	where $\underline{\x}_{-i} \in D_{-i}^{(v)}$, $\overline{\x}_{-i} \in D_{-i}^{(u)}$, and suitable $\vect t'_{-i}, \vect t_{-i}\in D_{-i}^{(u)}$ exist since $t$ ($t'$) is \sometime{0} (\sometime{1}) at $u$. By \eqref{eq:not-revealable-four-profiles} we get that the cycle $C:=\underline{\x} \rightarrow \vect t' \rightarrow \overline{\x} \rightarrow \vect t \rightarrow \underline{\x}$ has weight  $w(\underline{\x} \rightarrow \vect t') + w(\vect t' \rightarrow \overline{\x}) + w( \overline{\x} \rightarrow \vect t) + w(\vect t \rightarrow \underline{\x}) = 0 - t' +0+t <0$, thus contradicting OSP CMON. 
	Note that the edges of $C$ indeed belong to 
	$\ver$, since $\underline{\vect x}$ and $\overline{\vect x}$ are separated from $\vect t$ and $\vect t'$ at $v$ and $u$, respectively. This completes the proof.
\end{proof}

\subsubsection{The Sufficient Condition}
\begin{theorem}\label{thm:suffcmon}
An extremal mechanism with weak interleaving is OSP.
\end{theorem}
\begin{proof}
	Firstly observe that any such mechanism is OSP 2CMON by construction. Assume by contradiction that the mechanism is not OSP, and thus OSP CMON does not hold. Then there is some agent $i$ for which $\ver$ has a negative-weight cycle $C$ -- let $\bb 1, \bb 2, \bb 3$ and $ \bb 4$ be the profiles in $C$ identified by Theorem \ref{thm:anatomy}. 
	
	Let $u$ be the node of $\T$ where $\bbb 2$ and $\bbb 1$ are separated by the mechanism and $\bb 2_{-i}, \bb 1_{-i}  \in D_{-i}^{(u)}$. The existence of the edge $(\bb 2, \bb 1)$ implies that such a node $u$ exists. Similarly, let $v$ be the node of the tree where $\bbb 1$ and $\bbb 3$ are separated by $\T$ and $\bb 3_{-i}, \bb 1_{-i} \in D_{-i}^{(v)}$. Observe that $v$ must exist in the subtree comprised of the path from the root of the tree to $u$ and the tree rooted in $u$. In fact, if no such $v$ existed then the mechanism could not differentiate $\bb 1$ from $\bb 3$ and the existence of $(\bb 2, \bb 1)$ would yield the existence of $(\bb 2, \bb 3)$ -- a contradiction. Furthermore, the non-existence of $(\bb 2, \bb 3)$ implies that $u=v$, and that in the partition of $D_i^{(u)}$ defined at $u$, $\bbb 3$ and $\bbb 2$ belong to the same part. To see this, it is enough to note that were $v \neq u$ or $\bbb 3$ and $\bbb 2$ to belong to different parts, $\bb 2$ and $\bb 3$ would have to be connected by an edge. Therefore, at $u$ the mechanism splits $D_i^{(u)}$ into $\{\bbb 1\}$ and (possibly a superset of) $\{\bbb 2, \bbb 3\}$.
	
	We can now repeat the reasoning above on $\bb 3, \bb 4$ and $\bb 2$ and conclude that in $w$, the node of tree where the mechanism separates $\bbb 3$ from $\bbb 4$, $D_i^{(w)}$ is split into $\{\bbb 4\}$ and (possibly a superset of) $\{\bbb 2, \bbb 3\}$.
	
	Since both $u$ and $w$ have one branch labelled by $\bbb 2$ and $\bbb 3$ and that the leaves for profiles $\bb 2$ and $\bb 3$ will be in the subtree reached by that branch, it must be that $u$ is an ancestor of $w$ in $\T$ or viceversa. Focus on the first case, meaning that $\bbb 4 \in  D_i^{(u)}$, $\bb 4_{-i} \in D_{-i}^{(u)}$ as well. (The other case is similar.) But then, since $\bbb 1 < \bbb 2 < \bbb 3 < \bbb 4$ we have that $\bbb 1$ is the minimum at $u$. Similarly, we note that $\bbb 4$ is the maximum at $w$, which means that the mechanism has interleaving at $w$. Therefore, $i$ must have become revealable at some point in the history before (and including) $w$. But, since Condition~\ref{cond:revealable} of Lemma~\ref{le:local} holds (specifically, \eqref{eq:revealable} holds at $w$) it cannot be the case that $f_i(\bb 2)=0$ and $f_i(\bb 3)=1$ since $\bbb 2 < \bbb 3$, 
	a contradiction.
\end{proof}

\section{Two-way Greedy Algorithms}\label{sec:2way}
As noted above, our characterization does not require costs to be positive; so it holds true for negative costs, that is, valuations. We will then talk now simply about {type} when we refer to the agents' private information, be it costs or valuations. We will also sometimes use the terminology of in-query (out-query, respectively) to denote a bottom (top) query for costs and a top (bottom) query for valuations.
\begin{definition}[(Anti-)Monotone Functions]
We say that a function is monotone (antimonotone, respectively) in the type if it is decreasing (increasing, respectively) in the cost, and increasing (decreasing, respectively) in the valuation. 
\end{definition}

In this section, we translate our characterization into algorithmic insights on OSP mechanisms and show a connection between their format and a certainly family of adaptive priority algorithms that they use. As a by-product, we show the existence of a host of new mechanisms. We are able to provide the first set of upper bounds on the approximation guarantee of OSP mechanisms, that are independent from domain size (as in \cite{esa19,wine19}) or assumptions on the designer's power to catch and punish lies (as in \cite{ferraioli2017obvious,verification}), see Table \ref{tbl:upperbounds}. The table also contains the new bounds we can prove on the approximation of OSP mechanisms, by leveraging our algorithmic characterization (i.e., Theorems \ref{thm:unbss}--\ref{thm:knap}). 

\renewcommand*{\thefootnote}{\fnsymbol{footnote}}

\begin{table}
	\centering
	\begin{tabular}{|c||c|c|}
		\hline
		Problem & Bound \\
		\hline
		Known Single-Minded Combinatorial Auctions (CAs)   &  $\sqrt{m}$ (\cite{LOS} + Cor. \ref{cor:iaa})\\
		MST (\& weighted matroids) &  1  (\cite{K56} + Cor. \ref{cor:iaa})  \\
		Max Weighted Matching &  2  (\cite{Avi83} +  Cor. \ref{cor:iaa})  \\
		$p$-systems\tablefootnote{A $p$-system is a downward-closed set system $(E, \mathcal{F})$ where there are at most  $p$ \emph{circuits}, that is, minimal subsets of $E$ not belonging to $\mathcal{F}$ \cite{Hausmann1980}.} &  $p$  (\cite{Hausmann1980} + Cor. \ref{cor:iaa}) \\
		Weighted Vertex Cover & $2$  (\cite{Cla83} + Cor. \ref{cor:iaa}) \\
		Shortest Path & $\infty$ (Thm \ref{thm:unbss}) \\
		Restricted Knapsack Auctions & $\Omega(\sqrt{n})$ (Thm \ref{thm:sw:n:lb}) \\
		Asymmetric Restricted Knapsack Auctions ($3$ values) & $\sqrt{n-1}$ (Thm \ref{thm:sw:n:ub})\\
		Knapsack Auctions & $\Omega(\sqrt{\ln n})$ (Thm \ref{thm:knap}) \\
		\hline
	\end{tabular}
	\caption{Bounds on the approximation guarantee of OSP mechanisms. (The result for CAs has been also observed by \cite{bartmaria}. The mechanisms for matroids could be obtained also through Cor. \ref{cor:da}.)}
	\label{tbl:upperbounds}
\end{table}

\renewcommand*{\thefootnote}{\arabic{footnote}}
\setcounter{footnote}{1}

\paragraph{Immediate-Acceptance Auctions and Forward Greedy.} Let us begin by discussing forward greedy. 

\begin{definition}[Forward Greedy]
	A forward greedy algorithm uses functions $g_i^{(in)} : \mathbb{R} \times \mathbb{R}^{k} \rightarrow \mathbb{R}$, $k \leq n-1$, to rank the bids of the players and builds a solution by iteratively adding the agent with highest rank if that preserves feasibility (cf. Algorithm \ref{alg:fgreedy}).
	
	A forward greedy algorithm is monotone if each $g_i^{(in)}$ is monotone in $i$'s private type (i.e., its first argument). 
\end{definition}

\begin{algorithm}[tb]
	\caption{Forward greedy algorithm}\label{alg:fgreedy}
	Let $b_1, \ldots, b_n$ the input bids
	
	${\cal P} \leftarrow {\cal F}$ (${\cal P}$ is the set of all feasible solutions)
	
	${\cal I} \leftarrow \emptyset$  (${\cal I}$ is the set of infeasible agents, i.e., agents that cannot be included in the final solution)
	
	${\cal A} \leftarrow N$ (${\cal A}$ is the set of active or infeasible agents)
	
	\While{$|{\cal P}| > 1$}{
		
		Let $i = \arg\max_{k \in {\cal A} \setminus {\cal I}} g_k^{(in)}(b_k, \b_{N \setminus \mathcal{A}})$\label{l:order} 
		
		\If{there are solutions $S$ in $\cal P$ such that $i \in S$}
		{
			Drop from $\cal P$ all the solutions $S$ such that $i \not\in S$ (if any)
			
			${\cal A}  \leftarrow {\cal A} \setminus \{i\}$
			
		} \Else{
			${\cal I}  \leftarrow {\cal I} \cup \{i\}$
		}
	}
	Return the only solution in ${\cal P}$
\end{algorithm}

Few observations are in order for Algorithm \ref{alg:fgreedy}. Firstly, it is not too hard to see that it will always return a solution (i.e., there will eventually be a unique feasible solution in ${\cal P}$). To see this consider the case in which (at least) two solutions $S, S'$ are in ${\cal P}$. Then there must exist an agent $j$ such that $j \in S \Delta S'$, $\Delta$ denoting the symmetric difference between sets. This means that $j \in {\cal A} \setminus {\cal I}$ and the forward greedy algorithm will decide in the next steps whether $j$ is part of the solution or not (and consequently whether to keep $S$ or $S'$). 
Secondly, we stress how a forward greedy algorithm belong to the family of adaptive priority algorithms \cite{Borodin}. At Line \ref{l:order}, Algorithm \ref{alg:fgreedy} (potentially) updates the priority as a function of the bids of those bidders who have left the auction (i.e., that are not active anymore). A peculiarity of the algorithm (to do with its OSP implementation, cf. the proof of Corollary \ref{cor:iaa} below) is that the adaptivity does not depend on bidders who despite their high priority cannot be part of the eventual solution (i.e., the agents we add to ${\cal I}$).\footnote{Note that a syntactically (but not semantically) alternative definition of forward greedy algorithms could do without $\cal I$ by requiring an extra property on the priority functions (i.e., adaptively floor all the priorities of infeasible players).} This distinction would not be necessary if the problem at hand would be \emph{upward close} (i.e., if a solution $S$ is feasible then any $S'\supset S$ would be in $\mathcal{F}$ too). Thirdly, a notable subclass of forward greedy algorithms are fixed-priority algorithms, where Line \ref{l:order} and the while loop are swapped (and priority functions only depend on the agents' types). (These algorithms do not need to keep record of $\cal I$ either.) This algorithmic paradigm has been applied to different optimization problems (e.g., Kruskal's Minimum Spanning Tree (MST) algorithm \cite{K56}). 

We define \emph{Immediate-Acceptance Auctions (IAAs)} as mechanisms using only in-queries. 

\begin{corollary}\label{cor:iaa}
	 An IAA that uses algorithm $f$ is OSP if and only if $f$ is monotone forward greedy.
\end{corollary}
\begin{proof}
We will prove this in the case in which types are costs; the proof is dual for valuations. 

Let us start with the if condition; given a forward greedy algorithm $f$, we recursively define the implementation tree $\T$ of an IAA as follows. At each $u \in \T$, we compute $j = \arg\max_{i \in \mathcal{A} \setminus \mathcal{I}} g_i^{(in)}(\min D_i^{(u)}, \b_{N \setminus {\cal A}})$, where $D_i^{(u)}$ denotes, as above, the current domain of $i$ at $u$, $\b_{N \setminus {\cal A}}$ contains the bids of all the agents that have been queried and have replied yes to a preceding query, and $\cal I$ contains the bidders that could not be queried since they did not belong to any feasible solution. Since the algorithm is monotone, $j$ is the agent with maximum priority at this point of the execution. If including $j$ in the solution is feasible then we ask $j$ an in-query to check if her type is $\min D_j^{(u)}$. If she replies yes, we select her accordingly and continue the execution (i.e., remove $j$ from $\cal A$ and set $b_j=\min D_j^{(u)}$). If she replies no, we remove $\min D_j^{(u)}$ from her current domain and continue accordingly. If selecting $j$ is not feasible, then we add $j$ to $\cal I$ and continue the mechanism without any query to $j$. The mechanism then uses only a sequence of in-queries. Therefore, by our characterization, there exist payments $p$ such that $(f, p, \T)$ is OSP.

For the opposite direction, let $(f, p, \T)$ be an OSP IAA; by our characterization, we can restrict without loss of generality to the extremal implementation. We define a monotone priority function $g^{(in)}$ for $f$, thus showing that $f$ is indeed a monotone forward greedy algorithm. We will do that recursively by visiting $\T$ in level order. Given the root $r$ of $\T$ where we are in-querying $i$ we define the rank function $g_i^{(in)}(\min D_i, \emptyset)$ to be a suitable large value $V$.  (Recall that in-queries are made for the smallest possible cost in the domain.) Let $j$ denote the agent queried in the child $v$ of $r$ where $i$ has confirmed her type to be $\min D_i$. We set $g_j^{(in)}(\min D_j, \{\min D_i\})$ to be $V-1$. While for $k$, the agent interrogated at the other child of $r$, we set $g_j^{(in)}(\min D_j, \emptyset)=V-1$. Observe how we are gradually building the set $N\setminus {\cal A}$ of agents that have left the auction. (We remark how for this direction of the proof we do not need to use $\cal I$.) We can continue in this fashion until we reach the leaves of $\T$. 
\end{proof}
We note how all fixed-priority algorithms are OSP but not all OSP mechanism must use a fixed-priority algorithm. In fact, for a fixed-priority algorithm the sufficiency proof alone would go through; the second parameter of the priority functions is only needed for the opposite direction. 

It is important to compare the result above with \cite[Footnote 15]{MS20}, where it is observed how the strategic properties of forward and reverse greedy algorithms are different. The remark applies to auctions where these algorithms are augmented by the so-called threshold payment scheme. Our result shows that there exist alternative payment schemes for this important algorithmic design paradigm.~\footnote{It is indeed crucial to define payments that depend on the implementation tree for an IAA that is OSP; in Appendix \ref{sec:cas}, we give an example showing how this works for a well known forward-greedy algorithm for single-minded CAs.}

\paragraph{Deferred-Acceptance Auctions and Reverse Greedy.} Another algorithmic approach that can be used to obtain OSP mechanisms is reverse greedy (see Algorithm \ref{alg:bgreedy}), that is, having an out-priority function that is antimonotone with each agent's type and drops agents accordingly. In its auction format, this is known as Deferred-Acceptance Auctions (DAAs), see, e.g., \cite{MS20}. 

\begin{definition}[Reverse Greedy]
	A reverse greedy algorithm uses functions $g_i^{(out)} : \mathbb{R} \times \mathbb{R}^k \rightarrow \mathbb{R}$, $k \leq n-1$, to rank the bids of the players and iteratively excludes the player with highest rank (if feasible) until only one  solution is left (cf. Algorithm \ref{alg:bgreedy}). 
	
	A reverse greedy algorithm is antimonotone if each $g_i^{(out)}$ is antimonotone in $i$'s private type (i.e., its first argument).
\end{definition}

\begin{algorithm}[tb]
	\caption{Reverse greedy algorithm}\label{alg:bgreedy}
	Let $b_1, \ldots, b_n$ the input bids
	
	${\cal P} \leftarrow {\cal F}$ (${\cal P}$ is the set of all feasible solutions)
	
	${\cal I} \leftarrow \emptyset$ (${\cal I}$ is the set of in agents, i.e., those that cannot be dropped)
	
	${\cal A} \leftarrow N$ (${\cal A}$ is the set of active or in agents)
	
	\While{$|{\cal P}| > 1$}{
		
		Let $i = \arg\max_{k \in {\cal A} \setminus {\cal I}} g_k^{(out)}(b_k, \b_{N \setminus {\cal A}})$ 
		
		\If{there are solutions $S$ in $\cal P$ such that $i \not\in S$}
		{
			Drop from $\cal P$ all the solutions $S$ such that $i \in S$ (if any)
			
			${\cal A}  \leftarrow {\cal A} \setminus \{i\}$
		} \Else{
			${\cal I}  \leftarrow {\cal I} \cup \{i\}$
		}
	}
	Return the only solution in ${\cal P}$
\end{algorithm}
Similarly to the case of forward greedy, the algorithm need not use $\cal I$ for \emph{downward-closed} problems (i.e., every subset of a feasible solution is feasible as well). Incidentally, this is the way it is discussed in \cite{DGR17}.

The following corollary can be proved very similarly to Corollary \ref{cor:iaa} (in fact, it is sufficient to substitute in-queries and $\min$ with out-query and $\max$, respectively).

\begin{corollary}[\cite{MS20}]\label{cor:da}
A DAA using algorithm $f$ is OSP if and only if $f$ is an antimonotone reverse greedy algorithm. 
\end{corollary}

It is convenient to discuss the relative power of forward and reverse greedy algorithms. We now know from Corollaries \ref{cor:iaa} and \ref{cor:da} that they are strategically equivalent. Algorithmically, however, there are some differences. There are algorithms and problems, such as, the aforementioned Kruskal and MST, where we can take the reverse version of forward greedy (e.g., for MST,  start with the entire edge set, go through the edges from the most expensive to the cheaper, and remove an edge whenever it does not disconnect the graph) without any consequence to the approximation guarantee. For the minimum spanning tree problem (and, more generally, for finding the minimum-weight basis of a matroid), the reverse greedy algorithm is just as optimal as the forward one. In general (and even for, e.g., bipartite matching), the reverse version of a forward greedy algorithm with good approximation guarantee can be bad \cite{DGR17}. 

\paragraph{OSP Mechanisms and Two-way Greedy.} Our characterization shows that the right algorithmic technique for OSP is a suitable combination of forward and reverse greedy. 

\begin{algorithm}[tb]
	\caption{Two-way greedy algorithm}\label{alg:2greedy}
	Let $b_1, \ldots, b_n$ the input bids
	
	${\cal P} \leftarrow {\cal F}$ (${\cal P}$ is the set of all feasible solutions)
	
	${\cal I} \leftarrow \emptyset$ (${\cal I}$ is the set of infeasible and in agents)
	
	${\cal A} \leftarrow N$ (${\cal A}$ is the set of all agents that are active or infeasible or in)
	
	\While{$|{\cal P}| > 1$}{
		
		Let $i^{(in)} = \arg\max_{k \in {\cal A} \setminus {\cal I}} g_k^{(in)}(b_k, \b_{N \setminus {\cal A}})$
		
		Let $i^{(out)} = \arg\max_{k \in {\cal A} \setminus {\cal I}} g_k^{(out)}(b_k, \b_{N \setminus {\cal A}})$
		
		Let $i$ be the agent corresponding to $\max \left\{ g_{i^{(in)}}^{(in)}(b_{i^{(in)}}, \b_{N \setminus {\cal A}}), g_{i^{(out)}}^{(out)}(b_{i^{(out)}}, \b_{N \setminus {\cal A}}) \right\}$
		
		\If{$i = i^{(in)} \wedge$ there are solutions $S$ in $\cal P$ such that $i \in S$}{ 
			Drop from $\cal P$ all the solutions $S$ such that $i \not\in S$ (if any)\label{l:in}
			
			${\cal A} \leftarrow {\cal A} \setminus \{i\}$
		} 
		\ElseIf{$i = i^{(out)} \wedge$ there are solutions $S$ in $\cal P$ such that $i \not\in S$}{
			Drop from $\cal P$ all the solutions $S$ such that $i \in S$ (if any)\label{l:out}
			
			${\cal A} \leftarrow {\cal A} \setminus \{i\}$
		} \Else {
			${\cal I}  \leftarrow {\cal I} \cup \{i\}$
		}
	}
	Return the only solution in ${\cal P}$
\end{algorithm}

\begin{definition}[Two-way Greedy]
	A two-way greedy algorithm uses functions $g_i^{(in)} : \mathbb{R} \times \mathbb{R}^k \rightarrow \mathbb{R}$ and $g_i^{(out)} : \mathbb{R} \times \mathbb{R}^k \rightarrow \mathbb{R}$, $k \leq n-1$, for each agent $i$, to rank the bids of the players, and iteratively and greedily includes (if highest priority comes from a $g^{(in)}$ function) or excludes (if highest priority comes from a $g^{(out)}$) -- whenever possible -- the player with the highest rank until one feasible solution is left (cf. Algorithm \ref{alg:2greedy}). 
	
	A two-way greedy is all-monotone if each $g_i^{(in)}$ is monotone in $i$'s private type and $g_i^{(out)}$ is antimonotone in $i$'s private type (i.e., their first argument).
\end{definition}

To fully capture the strategic properties of two-way greedy algorithms, we need to define one more property for which we need some background definitions. Consider the total increasing\footnote{For notational simplicity, we here assume that there are not ties between the priority functions.} ordering $\varphi$ of the 
$2 \prod_i |D_i|$ functions\footnote{The algorithm must not necessarily have a definition for the priority functions for all the combinations of type/history as some might never get explored. In this case, we set  all the undefined entries to sufficiently small (tie-less, for simplicity) values.} $$\{g_i^{(\star)}(b, \b)\}_{i\in N, b \in D_i, \b \in D_{-i}, \star \in \{in, out\}}$$ used by a two-way greedy algorithm. Given $\varphi_\ell=g_i^{(\star)}(b_i, \b_{N\setminus \mathcal{A}_\ell})$, the $\ell$-th entry of $\varphi$, we let $D_j^{\prec}(\ell)$ ($D_j^{\succ}(\ell)$, respectively) denote the set of types $b \in D_j$ such that $g_j^{(\star)}(b, \b_{N \setminus {\cal A}}) > \varphi_\ell$ ($g_j^{(\star)}(b, \b_{N\setminus \mathcal{A}}) < \varphi_\ell$, respectively) with $\mathcal{A} \supseteq \mathcal{A}_\ell$ ($\mathcal{A} \subseteq \mathcal{A}_\ell$, respectively). Moreover, we add $b_i$ (the bid defining $\varphi_\ell$) to $D_i^\prec(\ell)$. In words, once in the ordering we reach the $\ell$-th entry for a certain agent type and a given ``history'' (i.e., $\b_{N \setminus {\cal A}_\ell}$) with $D_j^{\prec}(\ell)$ we denote all the types that the algorithm has already explored for agent $j$ at this point (that is, for a compatible prior history $\b_{N \setminus {\cal A}_\ell}$ with $\mathcal{A} \supseteq \mathcal{A}_\ell$). Similarly, $D_j^{\succ}(\ell)$ denotes those types in $D_j$ that are yet to be considered from this history onwards. Finally, for $d \in \{in, out\}$ we let $\overline{d}$ be a shorthand for the other direction, that is, $\{\overline{d}\}=\{in, out\} \setminus \{d\}$.

\begin{definition}[Interleaving Algorithm]\label{def:int}
We say that a two-way greedy algorithm is interleaving if for each $i$ and $\mathcal{A} \subseteq N$ the following occurs. For each $\varphi_{\ell}=g_i^{(d)}(b, \b_{N \setminus \mathcal{A}})$ such that for some $\mathcal{A'} \subseteq \mathcal{A}$ it holds $g_i^{(\overline{d})}(b', \b_{N \setminus {\cal A'}})=\varphi_{\ell'}$ with $b' \in D_i^{\succ}(\ell)$ (and then $\ell' > \ell$) we have
\[
g_i^{(\star)}(x, \b_{N \setminus \mathcal{A'}}) > g_j^{(\star)}(y, \b_{N \setminus \mathcal{A'}})
\]
for each $y \in D_j^\succ(\ell')$ and for all (but at most one) $x$ in $D_i^\succ(\ell')$ ($\star \in \{in, out\}$).
\end{definition}

The definition above captures in algorithmic terms the weak interleaving property of extensive-form implementations. Whenever there is a change of direction (from $d$ to $\overline{d}$) for a certain agent $i$ and two compatible histories (cf. condition $\mathcal{A'} \subseteq \mathcal{A}$) then it must be the case that $i$ is revealable and all the other unexplored types (but at most one) \emph{must} be explored next. 

\newcommand{\tmin}{t_{\min}}
\newcommand{\tmax}{t_{\max}}
\newcommand{\tmed}{t_{\med}}

\begin{example}[Interleaving Algorithm]
Consider a setting with three agents, called $x$, $y$ and $z$. 
The \emph{valuation} domain  is the same for all the agents and has maximum $\tmax$ and minimum $\tmin$. Consider the two-way greedy algorithm with the following ordering $\varphi$:
\[
g^{(in)}_x(\tmax) > g^{(in)}_y(\tmax) > g^{(out)}_z(\tmin) > g^{(out)}_y(\tmin) > \ldots
\]
(where the second argument is omitted since it is $\emptyset$). Let us focus on $\varphi_2=g^{(in)}_y(\tmax)$. Here $D_x^\prec(2)$ and $D_y^\prec(2)$ is $\{\tmax\}$ (as it has been already considered for both agents) whilst $D_x^\succ(2)$ and $D_y^\succ(2)$ is equal to the original domain but $\tmax$. For $z$, instead, $D_z^\prec(2)=\emptyset$ and $D_z^\succ(2)$ is still the original domain. At $\varphi_4$ there is a change of direction for agent $y$. The two-way greedy is interleaving if the domain has only three types (since there are no constraints on the in/out priority for third type in the domain of $y$). For larger domains, instead, we need to look at the next entries of $\varphi$ to ascertain whether the algorithm is interleaving or not. 
\end{example}

We are now ready to show that two-way greedy is indeed the algorithmic nature of OSP. 
\begin{corollary}
	A mechanism using algorithm $f$ is OSP if and only if $f$ is an all-monotone interleaving two-way greedy algorithm.
\end{corollary}
\begin{proof}
	The result can be proved very similarly to Corollaries \ref{cor:iaa} and \ref{cor:da}. For the if part, we can define an implementation tree using the priority functions exactly as described in those proofs.
	
	For the only if direction, we cannot directly use the implementation tree of the OSP mechanism to iteratively define the priority functions, as this might violate the interleaving property (in general an OSP mechanism can perform the revelation queries in any order rather than immediately as requested by Definition \ref{def:int}). However, we can use the claim below and first modify the OSP mechanism to be \emph{well ordered}. We say that a mechanism is {well ordered} if whenever an agent becomes revealable, the mechanism finds her type before querying any other agent. We next show (proof in the appendix) that we can restrict our attention to extremal OSP mechanisms that are well ordered.
	
	\begin{observation}\label{obs:wellordered}
		For any OSP extremal mechanism $\M=(f, p, \T)$, there is a well-ordered extremal mechanism $\M'=(f, p, \T')$ that is OSP.
	\end{observation}
	We can now use the new implementation tree and obtain priority functions that define a two-way greedy algorithm that is all-monotone and interleaving. This concludes the proof.
\end{proof}

Given the corollary above, we will henceforth simply say two-way greedy (algorithm) and avoid stating the properties of all-monotonicity and interleaving. The next subsections contain an analysis of two-way greedy in terms of approximation guarantee (in comparison to forward/reverse in Section \ref{sec:2way-1way} and in their own right in Section \ref{sec:2way-apx}).

\subsection{Two-way vs Forward/Reverse Greedy}\label{sec:2way-1way}
We next prove that two-way greedy algorithms are in general more powerful than single-directional greedy algorithms, for both social welfare maximization (types are valuations) and social cost minimization (types are costs).

We begin with costs and prove the following theorem.
\begin{theorem}[Social Cost]\label{thm:1vs2:sc}
	There are set systems in which two-way greedy algorithms are $\alpha$-approximate, for some $\alpha > 1$, but both forward and reverse greedy do not return a better than $\rho$-approximation  to the optimal social cost, with $\rho/\alpha>1.02$.
\end{theorem}
\begin{proof}
	Consider two parallel solutions  with $1$ and $2$ agents, respectively: $S=\{x\}$ and $T=\{y,z\}$  and let their domain be 
	$D_x = D_y = D_z = \{\tmin,\tmed,\tmax\}= \{10, 22, 36\}$. 
	We shall consider the first query of the algorithm (thus $\b_{N \setminus {\cal A}}= \b_{N\setminus N}={\emptyset}$) to some agent $a$, and use  $g_a^{(\star)}(b_a)$ as a shorthand for $g_a^{(\star)}(b_a, {\emptyset})$ to represent this first query ($\star \in \{in, out\}$ and $a\in\{x,y,z\}$). We will denote as $(b_x, b_y, b_z)$ an instance of the problem in which the type of the unique agent in $S$ is $b_x$ and agents in $T$ have types $b_y$ and $b_z$. We distinguish the two cases and the corresponding approximation guarantee $\beta$:
	\begin{description}
		\item[Forward greedy $\left(\beta \geq \min\left\{\frac{\tmax}{2\tmin},1+\frac{\tmin}{\tmax}\right\}\right)$:] Consider $g_S = g_x^{(in)}(36)$ and $g_T= \max\{g_y^{(in)}(10), g_z^{(in)}(10)\}$. If $g_S > g_T$, it means that the algorithm has approximation at least $\frac{36}{20} > 1.1$ (consider the instance $(36, 10, 10)$). On the other hand, if  $g_S < g_T$ then the approximation guarantee of the algorithm is $\frac{46}{36} > 1.1$ (consider the instance $(36, 10, 36)$). 
		
		\item[Reverse greedy $\left(\beta \geq \min\left\{\frac{2\tmed}{\tmax},\frac{\tmax}{\tmin+\tmed}\right\}\right)$:] Let $g_S = g_x^{(out)}(36)$ and $g_T= \max\{g_y^{(out)}(22),g_z^{(out)}(22)\}$. If $g_S > g_T$, it means that the the algorithm is $\frac{44}{36} > 1.1$ approximate (consider the instance $(36, 22, 22)$). On the contrary, if $g_S < g_T$, then the approximation guarantee of the algorithm is at least $\frac{36}{32} > 1.1$ (consider the instance $(36, 10, 22)$, where $22$ is for the agent in $T$ with maximum rank). 
	\end{description}
	Intuitively, in either case (forward or reverse) greedy has to decide too early (which agent will be in) the solution.
	
	Consider now the following algorithm: return $T$ if at least one of its agents has type $10$, no of its agents have type $36$, and the agent in $S$ has type $36$, otherwise return $S$.
	It is not hard to see that this algorithm can be implemented as a two-way greedy algorithm. The only case in which this algorithm does not return an optimal outcome is on instance $(22, 10, 10)$, where the optimum would be $T$, and the algorithm returns $S$. However, the approximation ratio is $\frac{22}{20} = 1.1$.
\end{proof}

For valuations, we have the following bound.
\begin{theorem}[Social Welfare]\label{thm:1vs2:sw}
 There are set system problems in which two-way greedy algorithms are optimal whilst both forward and reverse greedy do not return a better than $\sqrt{2}$-approximation to the optimal social welfare.
\end{theorem}
\begin{proof}
Consider the case in which $\mathcal{F}$ is comprised of only two solutions $S=\{x\}$ and $T=\{y,z\}$ with $1$ and $2$ agents, respectively, and let the \emph{valuation} domain be $D_i = \{\tmin,\tmed,\tmax\}$ for each $i$, where $\tmin=0$, $\tmed=1/\sqrt{2}$ and $\tmax=1$. Observe that in this setup as soon as an agent becomes inactive the outcome is determined; therefore, the priority functions will never use their second parameter. Accordingly, as above, we use here $g^{(\star)}(b)$ as a shorthand for $g^{(\star)}(b, \emptyset)$. As above, $(b_x, b_y, b_z)$ denotes an instance in which the type of the unique agent in $S$ is $b_x$ and agents in $T$ have types $b_y$ and $b_z$. 

We begin with the analysis of a forward greedy algorithm. We next show that this algorithm will either return $S$ for the profile $(\tmax, \tmed, \tmed)$ or $T$ for the profile $(\tmax, \tmed, \tmin)$. In both cases, the approximation ratio is $\sqrt{2}$. We begin by observing that the priority of $y$ and $z$ for $\tmax$ must be the two highest ones (otherwise the approximation guarantee would be at least $1+1/\sqrt{2} > \sqrt{2}$). Now it can either be $g_x^{(in)}(\tmax)>\max\{g_y^{(in)}(\tmed), g_z^{(in)}(\tmed)\}$ -- which corresponds to the former outcome -- or $g_x^{(in)}(\tmax)<\max\{g_y^{(in)}(\tmed), g_z^{(in)}(\tmed)\}$ -- which implies the second outcome. 

Similarly, a reverse greedy either returns $S$ for the profile $(\tmed, \tmax, \tmin)$ or $T$ for $(\tmed, \tmin, \tmin)$. The approximation ratio is at least $\sqrt{2}$ in both cases. Note that $g_x^{(out)}(\tmin)$ must be the highest priority (for otherwise the approximation is worse than $\sqrt{2}$). Then we compare 
$g_x^{(out)}(\tmed)$ with $\max\{g_y^{(out)}(\tmin), g_z^{(out)}(\tmin)\}$. If the latter is larger than the former we get the first outcome; otherwise, we get the second outcome.

We now conclude the proof by observing that the following algorithm is optimal on this instance: return $T$ if $b_y = \tmax$ or $b_z = \tmax$ or $b_x = \tmin$; otherwise, if $b_y = \tmin$ or $b_z = \tmin$ return $S$ and, in all the remaining cases, return $T$. It is not hard to see that this can be actually implemented with a two-way greedy algorithm. Specifically, the the following priority functions would suffice:
\[
g_y^{(in)}(\tmax)>g_z^{(in)}(\tmax)>g_x^{(out)}(\tmin)>g_y^{(out)}(\tmin)>g_z^{(out)}(\tmin)>g_y^{(in)}(\tmed).\qedhere
\]
\end{proof}

Finally, we consider the relative power of the greedy paradigms for downward-closed set systems (and social welfare maximization).

\begin{theorem}[Downward-Close Solutions]\label{thm:1vs2:dc}
	There are downward-close set system problems in which two-way greedy algorithms can return an $\alpha$-approximation, whilst both forward and reverse greedy do not return a better than $\beta$-approximation to the optimal social welfare, with $\frac{\beta}{\alpha} \approx 1.3$.
\end{theorem}
\begin{proof}
	Consider the case in which $\mathcal{F}$ is comprised of solutions $S=\{i_S\}$, $T=\{i_1, \ldots, i_k\}$ with even $k\geq 4$, and all the subsets of $T$. Let the valuation domain be $D_i = \{\tmin,\tmed,\tmax\}$ for each $i$, where $\tmin=1$, $\tmed=\rho k$ and $\tmax=\frac{k}{2}\tmin + \frac{k}{2}\tmed = \frac{k}{2} (\rho k + 1)$, with $\rho = \frac{1+\sqrt{5}}{2}$.
	
	Consider then the algorithm defined in Algorithm \ref{alg:dc}. 
	\begin{algorithm}[tb]
		\SetKw{Break}{break}
		\caption{Two-greedy Algorithm for the instance of Theorem~\ref{thm:1vs2:dc}}\label{alg:dc}
		Set $R=\emptyset$ \tcp*[h]{R records the $i \in T$ that cannot be returned}\;
		\lIf{there is $i \in T$ whose valuation is $\tmax$}{\Return $T$} \nllabel{l1}
		\lIf{$i_S$ has valuation $\tmin$}{\Return $T$} \nllabel{l2}
		\For{$j = 1, \ldots, k$}{
			\If{$i_j$ has valuation $\tmin$}{
				Set $R = R \cup \{i_j\}$\;
				\lIf{$|R| = \frac{k}{2}$}{\Break}
			}\Else{
				\lIf{$i_S$ has valuation $\tmed$}{\Return $T\setminus R$} \nllabel{l3}
			}
		}
		\If{$|R|=\frac{k}{2}$}{
			\lIf{$i_S$ has valuation $\tmax$}{\Return $S$} \nllabel{l4}
			\If(\tcp*[h]{I.e., it has not been checked if $i_S$ has valuation $\tmed$}){$j = \frac{k}{2}$}{
				\lIf{there is $i \in T\setminus R$ whose valuation is $\tmed$}{\Return $T\setminus R$} \nllabel{l5}
				\lElse{\Return $S$} \nllabel{l6}
			}
		}\lElse{\Return $T \setminus R$} \nllabel{l7}
	\end{algorithm}
	It is not hard to check that this algorithm is two-way greedy (incidentally, 
	the interleaving property is trivially satisfied by these rankings since the valuation domain has size 3).
	
	As for the approximation ratio of this algorithm, observe that it returns the optimal solution at Line~\ref{l1}, Line~\ref{l3}, Line~\ref{l4}, Line~\ref{l6}. Instead, if the solution returned at Line~\ref{l3}, Line~\ref{l5} or Line~\ref{l7} has value $C$, then the optimal solution would have value $C + |R|\tmin$, with $C \geq \tmed + (k - |R|-1)\tmin$ and $|R| \leq \frac{k}{2}$. Hence, the approximation ratio of this algorithm is
	\begin{align*}
	\alpha & = \frac{C + |R|\tmin}{C} = 1 + \frac{|R|\tmin}{C} \leq 1 + \frac{|R|\tmin}{\tmed + (k - |R|-1)\tmin}\\
	& \leq 1 + \frac{\frac{k}{2}\tmin}{\tmed + \left(\frac{k}{2}-1\right)\tmin} = 1 + \frac{\frac{k}{2}}{\rho k + \frac{k}{2} - 1} = \frac{(\rho+1)k-1}{\left(\rho+\frac{1}{2}\right)k - 1}.
	\end{align*}
	
	Let us now consider reverse greedy algorithms. After that solution $S$ has been ruled out if $i_S$ has valuation $\tmin$, the algorithm either must rule out all $i \in T$ if their valuation is $\tmin$ or it rules out $S$ if $i_S$ has valuation $\tmed$. Indeed, if none of these cases occurs, then the same solution is returned on instance $\b$ such that $b_{i_S} = \tmax$, and $b_i = \tmin$ for every $i \in T$ and on instance $\b'$ such that $b'_{i_S} = \tmed$, $b'_{i^*} = \tmax$ and $b'_i = \tmin$ for every $i \neq {i^*} \in T$, where $i^*$ is the one $i \in T$ that is not ruled out if the valuation is $\tmin$. In any case, the returned solution will have approximation larger than $\frac{k}{2} > \rho$.
	
	If there is $i^*$ such that $i^*$ is not ruled out if her valuation is $\tmin$ before  $i_S$ is ruled out if her valuation is $\tmed$, 
	then it must be the case that on instance $\b'$ such that $b'_{i_S} = \tmed$ and $b'_{i} = \tmin$ for every $i \in T$, the returned solution will have a value that is at most $k\tmin$. Hence, the approximation ratio would be $\frac{\tmed}{k\tmin} = \rho$.
	
	If, instead, all $i \in S$ are ruled out if their valuation is $\tmin$ before  $i_S$ is ruled out if her valuation is $\tmed$, 
	then it must be the case that on instance $\b'$ such that $b'_{i_S} = b'_{i^*} = \tmed$ and $b'_{i} = \tmin$ for every $i \neq i^* \in T$, $i^*$ being the last $i \in T$ to be processed by the algorithm, 
	the returned solution will have a value that is at most $\tmed$, whereas the optimal solution has value $\tmed + (k-1)\tmin$. Hence, the approximation ratio would be $\beta = \frac{\tmed + (k-1)\tmin}{\tmed} = \frac{(\rho+1)k-1}{\rho k}$.
	Observe that $\beta \leq \frac{\rho+1}{\rho} = \frac{3+\sqrt{5}}{1+\sqrt{5}} = \rho$.
	
	Consider now forward greedy algorithms. After the solution $T$ is returned if there is $i \in T$ whose valuation is $\tmax$, the algorithm must either return $S$ if the valuation of $i_S = \tmax$ or return $T$ if $i$ whose valuation is $\tmed$, for some $i \in T$. Indeed, if none of these case occurs, then the same solution is returned on instance $\b'$ such that $b'_{i_S} = \tmax$ and $b'_i = \tmin$ for every $i \in T$ and on instance $\b''$ such that $b''_{i_S} = \tmin$, and $b''_i = \tmed$ for every $i \in T$. In both cases, the returned solution will have approximation ratio larger than $\rho$.
	
	If the algorithm return $S$ if $i_S$ has valuation $\tmax$ before that some $i \in T$ is evaluated, 
	then it must be the case that the algorithm returns $S$ on instance $\b'$ such that $b'_{i_S} = \tmax$ and $b'_i = \tmed$ for every $i \in T$, while $T$ is the optimal solution. Hence, the approximation ratio would be $\frac{k\tmed}{\tmax} = \frac{\rho k}{\frac{\rho k^2}{2} + \frac{k}{2}} = \rho \frac{2 k}{\rho k - 1} \geq \rho$.
	
	Finally, if the algorithm return $S$ as soon as there is $i \in T$ whose valuation is $\tmed$, 
	then the algorithm returns $T$ on instance $\b'$ such that $b'_{i_S} = \tmax$, $b'_{i^*} = \tmed$ and $b'_i = \tmin$ for every $i \neq i^* \in T$. Hence, in this case the approximation is larger than $\frac{k}{2} > \rho$.
	
	In conclusion, no forward and reverse algorithm can return an approximation that is better than $\beta$. Hence the ratio between the approximation ratio of the two-way greedy algorithm and the best forward/reverse greedy algorithm is $\frac{\beta}{\alpha} = \frac{\left(\rho + \frac{1}{2}\right) k - 1}{\rho k} = \frac{(2+\sqrt{5})k-1}{(1+\sqrt{5})k}$, that tends to $\frac{2+\sqrt{5}}{1+\sqrt{5}} \approx 1.3$ as $k$ increases.
\end{proof}

\subsection{Approximation Guarantee of Two-way Greedy Algorithms}\label{sec:2way-apx}
We next prove that the approximation ratio of two-way greedy algorithms is unbounded in set systems where we want to output the solution with minimum social cost \emph{for some} (minimal) structure of $\mathcal{F}$. Examples include the shortest path problem. For such problems, a two-way greedy algorithm must commit immediately to a solution after the first decision in either Line \ref{l:in} or \ref{l:out}. This leads to unboundedness.

\begin{theorem}[Social Cost]\label{thm:unbss}
 For any $\rho \geq 1$, there exists a set system such that no two-way greedy algorithm returns a $\rho$-approximation to the optimal social cost, even if there are only two feasible solutions and four agents.
\end{theorem}
\begin{proof}
 Consider a set system with only two feasible solutions, say $S$ and $T$, consisting of $s > 1$ and $k \geq s$ agents, respectively, and no agent belonging to both solutions.
 For each agent $i$, we suppose that the cost domain has minimum $\tmin$ and maximum $\tmax$, such that $\frac{\tmax}{\tmin} > \max\left\{\frac{\rho k-1}{s-1}, \rho \cdot \frac{k-1}{s-\rho}\right\} \geq \frac{k-1}{s-1}$. We focus on the maximum rank priority at the beginning of the algorithm (and, accordingly, we drop the second parameter for the priority functions). 
 
If the maximum is $g_i^{(in)}(\tmin)$ for $i \in S$ then it may be the case that the returned solution has cost $\tmin + (s-1)\tmax$, whereas the remaining solution would have cost $k \tmin < \tmin + (s-1)\tmax$. The approximation factor would then be at least $\frac{\tmin+(s-1)\tmax}{k \tmin} > \rho$.
 
 On the other hand, if the maximum is  $g_i^{(out)}(\tmax)$ for $i \in T$ then it may be the case that the returned solution has cost $s\tmax$, whereas the remaining solution would have cost $\tmax+(k-1)\tmin < s\tmax$. The approximation factor would then be at least $\frac{s\tmax}{\tmax+(k-1)\tmin} > \rho$.
 
 It is easy to check that any other choice for the highest priority would be even worse than the ones discussed above (since we would be committing to output the solution with more agents or exclude the solution with less agents). The theorem then follows.
\end{proof}

We now consider the case where the players have valuations and we are interested in maximizing the social welfare. We call this setup where no further assumption on the structure of the feasible solutions in $\mathcal{F}$ can be made, a \emph{restricted knapsack auction} problem. This is very much related to the knapsack auctions in \cite{DGR17}, where the authors further assume that $\mathcal{F}$ is downward closed (see below for a discussion of this setup). The theorem below proves that when $\mathcal{F}$ is comprised of only two different solutions, then a two-way greedy algorithm cannot return a better than $\sqrt{n-1}$-approximation.
\begin{theorem}[Social Welfare]\label{thm:sw:n:lb}
	There exists a set system for which any two-way greedy algorithm has approximation $\Omega(\sqrt{n})$ to the optimal social welfare, even if there are only two feasible solutions.
\end{theorem}
\begin{proof}
Consider the case in which $\mathcal{F}$ is comprised of only two solutions $S=\{i_S\}$ and $|T|=k=n-1>1$ with $S \cap T = \emptyset$. Let the \emph{valuation} domain be $D_i = \{\tmin,\tmed,\tmax\}$ for each $i$, where $\tmin=0$, $\tmed=1/\sqrt{k}$ and $\tmax=1$. As in the proof of Theorem \ref{thm:1vs2:sw}, we can here restrict, without loss of generality, to priority functions that only use the first parameter (and, consequently, we drop the second parameter from our notation). We will prove that no two-way greedy algorithm can return a better than $\sqrt{k}$-approximation to the social welfare, thus proving the claim.

Assume not and let $\rho$ denote the approximation guarantee of the algorithm. We begin by showing some properties on the priority function. We have
\begin{equation}\label{eq:top:in}
g_{i_T}^{(in)}(\tmax) > \max \left\{\max_{j\in T} g_j^{(out)}(\tmin), g^{(in)}_{i_S}(\tmax), g^{(out)}_{i_S}(\tmed)\right\},
\end{equation}
where $i_T = \arg\min_{i \in T} g_{i}^{(in)}(\tmax)$ and, as from above, $i_S$ denotes the only agent in $S$. It is not hard to check that $\rho$ would be worse than $\sqrt{k}$ if the above were not satisfied. Specifically, by all-monotonicity, if for some $j \in T$ we had
\begin{align*}
g_{i_T}^{(in)}(\tmax) &< g_j^{(out)}(\tmin) \Longrightarrow \rho=\infty \; \text{(instance s.t. $b_j=\tmin$, $b_{i_S}=\tmin$ and $b_{i_T}=\tmed$ for all $T \ni j' \neq j$)},\\
g_{i_T}^{(in)}(\tmax) &< g_{i_S}^{(in)}(\tmax) \Longrightarrow \rho > \sqrt{k} \; \text{(instance s.t. $b_j=\tmed$ for all $T \ni j \neq i_T$, $b_{i_S}=\tmax$ and $b_{i_T}=\tmax$)},\\
g_{i_T}^{(in)}(\tmax) &< g^{(out)}_{i_S}(\tmed) \Longrightarrow \rho=\infty \; \text{(instance s.t. $b_j=\tmin$ for all $j \in T$, $b_{i_S}=\tmed$)}.
\end{align*}
Observe that by all-monotonicity, the inequality above leaves only the relative rank of $g_{i_S}^{(out)}(\tmin)$ undetermined. Now we compare 
\[
g_{j_T}^{(in)}(\tmed) \qquad \text{ with } \qquad g_{i_S}^{(out)}(\tmin),
\]
where $j_T = \arg\max_{j \in T} g_{j}^{(in)}(\tmed)$.

Assume first that $g_{j_T}^{(in)}(\tmed) > g_{i_S}^{(out)}(\tmin)$. Here the approximation guarantee is at least $\sqrt{k}$. Indeed, if $g_{i_S}^{(in)}(\tmax)>g_{j_T}^{(in)}(\tmed)$ consider the instance where $b_{i_S}=\tmax$ and $b_j=\tmed$ for all $j \in T$. For the opposite case, the bound is proved for the instance wherein $b_{i_S}=\tmax$, $b_j=\tmin$ for all $j\in T, j \neq j_T$, and $b_{j_T}=\tmed$.

Consider now the case $g_{j_T}^{(in)}(\tmed) < g_{i_S}^{(out)}(\tmin)$. We here have
\begin{equation}\label{eq:bottom:out}
g_{i_S}^{(out)}(\tmin) > \max \left\{\max_{j\in T} g_j^{(out)}(\tmin), g^{(in)}_{i_S}(\tmax)\right\}.
\end{equation}
If this were not true then the approximation guarantee $\rho$ would be worse than $\sqrt{k}$:
\begin{align*}
g_{i_S}^{(out)}(\tmin) &< g_j^{(out)}(\tmin) \Longrightarrow \rho=\infty \; \text{(instance where $b_{i_S}=\tmin$ and $b_{j}=\tmed$ for all $j \in S$)},\\
g_{i_S}^{(out)}(\tmin) &< g_{i_S}^{(in)}(\tmax) \Longrightarrow \rho \geq \sqrt{k} \; \text{(instance where $b_j=\tmed$ for all $j \in S$, $b_{i_S}=\tmax$)}.
\end{align*}
By putting together \eqref{eq:top:in} and \eqref{eq:bottom:out} we know that in-priorities for players in $T$ and out-priority for $i_S$ have highest rank. The proof now concludes by arguing that no matter what priority is next in rank, the algorithm will have an approximation ratio of at least $\sqrt{k}$:
\begin{align*}
g_{i}^{(in)}(\tmed), i \in T & \Longrightarrow \rho \geq \sqrt{k} \; \text{(instance s.t. $b_i=\tmed$, $b_{j}=\tmin$, for all $T \ni j \neq i$ and $b_{i_S}=\tmax$)},\\
g_{i}^{(out)}(\tmin), i \in T & \Longrightarrow \rho \geq k-1 \; \text{(instance s.t. $b_i=\tmin$, $b_j = \tmed$ for all $T \ni j \neq i$ and $b_{i_S}=\tmed$)},\\
g_{i_S}^{(in)}(\tmax) & \Longrightarrow \rho \geq \sqrt{k} \; \text{(instance s.t. $b_j=\tmed$ for all $j \in T$ and $b_{i_S}=\tmax$)},\\
g_{i_S}^{(out)}(\tmed) & \Longrightarrow \rho=\infty \; \text{(instance s.t. $b_j=\tmin$ for all $j \in T$, $b_{i_S}=\tmed$)}.\qedhere
\end{align*}
\end{proof}
Interestingly the proof above uses a so-called \emph{asymmetric instance} \cite{DGR17} where one solution is a singleton and the other is comprised of all the remaining bidders. We here prove that the analysis above is tight at least for three-value domains, by giving a matching upper bound for the asymmetric instances of restricted knapsack auctions. We design a forward greedy algorithm, thus also proving that there is no gap between forward and two-way greedy in this context.
\begin{theorem}\label{thm:sw:n:ub}
	There is a $\sqrt{n-1}$-approximate forward greedy algorithm for the asymmetric restricted knapsack auctions, when bidders have a three-value domain $\{\tmin, \tmed, \tmax\}$.
\end{theorem}
\begin{proof}
Let $k=n-1$ and let $T$ denote the solution with more than one agent in the asymmetric instance of the restricted knapsack auction, $S$ denoting the other (and $i_S$ its only element). The algorithm uses the priority function rank defined in Algorithm \ref{alg:sw:n:ub}. 
\begin{algorithm}[tb]
	\caption{Priorities for asymmetric instances of restricted knapsack auctions}\label{alg:sw:n:ub}
	
	Set $g_i^{(in)}(\tmax) > V$, for $i \in T$ and a suitably large $V$ 
	
	\If{$\tmed \leq \frac{\tmax}{k} \vee \frac{\tmax}{\tmed +(k-1) \tmin} > \frac{k \tmed}{\tmax}$}{ 
			Set $V' < g_i^{(in)}(\tmed) < g_{i_S}^{(in)}(\tmax) < V$, for all $i \in T$ and a suitably large $V'$\label{l:maxfirst}
	} \Else {
			Set $V > g_i^{(in)}(\tmed) > g_{i_S}^{(in)}(\tmax) > V'$, for all $i \in T$ and a suitably large $V'$\label{l:medfirst}
	}
	\If{$\tmed > k \tmin$}{
		Set $V'' < g_{i_S}^{(in)}(\tmed) < V'$ for a suitably large $V''$
    }
	Set $g_{i}^{(in)}(\tmin) < V''$ for an $i \in T$ (for a suitable $V'' < V'$ if undefined)
\end{algorithm}
In words, the algorithm gives maximum priority to the agents in $T$ for the highest possible valuation. Then depending on the relationship between the values in the domain, it either gives higher priority to the elements of $T$ for $\tmed$ or to $i_S$ for $\tmax$. Finally, if $\tmed > k \tmin$, the priority for $S$ and $\tmed$ is defined, before concluding with a ``dummy'' priority to return $T$. We next show that the approximation guarantee of the algorithm, denoted as $\rho$ below, is at most $\sqrt{k}$. 

First let us consider the case in which $\tmed \leq \tmax/k$. Here irrespectively of the value $\tmin$, the forward greedy algorithm defined above returns the optimal solution. 

When $\tmed > \tmax/k$ then it is not hard to see that $$\rho = \min \left\{\frac{\tmax}{\tmed +(k-1) \tmin}, \frac{k \tmed}{\tmax}\right\}.$$ The former (latter, respectively) occurs when we return $T$ ($S$, respectively) by giving a higher (lower, respectively) value to $g_i^{(in)}(\tmed)$, for $i \in T$, as opposed to $g_{i_S}^{(in)}(\tmax)$, cf. Line \ref{l:medfirst} (Line \ref{l:maxfirst}, respectively) of Algorithm \ref{alg:sw:n:ub}. Since $\tmin \geq 0$, then $\rho$ is maximized when 
\[
\frac{\tmax}{\tmed +(k-1) \tmin}=\frac{k \tmed}{\tmax} \Longrightarrow \tmed = \frac{\tmax}{\sqrt{k}},
\]
which yields $\rho=\sqrt{k}$ as desired.
\end{proof}

We now turn our attention to downward-close set systems for social welfare maximization. As from above, this is a generalization of the setting studied in \cite{DGR17}, called \emph{knapsack auctions }and defined as follows. There are $n$ bidders and $m$ copies of one item; each bidder has a private valuation $v_i$ to receive at least $s_i$ copies of the item, $s_i$ being public knowledge. A solution is feasible if the sum of items allocated to bidders is at most $m$. The objective is social welfare maximization. Recall that the authors give a $O(\ln m)$-approximate DAA/reverse greedy algorithm and prove a lower bound of $\ln^\tau m$, for a positive constant $\tau$, \emph{limited to} DAA/reverse greedy. We next show that the upper bound is basically tight for the entire class of OSP mechanisms. 
\begin{theorem}[Social Welfare Downward-Closed Set Systems]\label{thm:knap}
	There is a downward-close set system for which every two-way greedy algorithm has approximation $\Omega(\sqrt{\ln n})$ to the optimal social welfare.
\end{theorem}
\begin{proof}
	Consider the case in which $\mathcal{F}$ is comprised of the following solutions: $S=\{i_S\}$, $T$ -- with $|T|=k=n-1>1$ and $S \cap T = \emptyset$ -- and all the subsets of $T$. Let the valuation domain be $D_i = \left\{t_x \colon x = 1, \ldots, k\right\} \cup \{0, 1\}$ for each agent $i$, where $t_x=\frac{1}{x \sqrt{2\ln k}}$. From now on, let also set $\rho=\frac{\sqrt{\ln k}}{\sqrt{2}+\varepsilon}$, for $\varepsilon > 0$; $\rho$ will denote the approximation guarantee of the two-way greedy algorithm.
	
	Let us first discuss some properties of these valuations. First observe that for every $x \geq 2\lceil\ln k\rceil$
	\begin{equation}\label{eq:M1}
	x t_2 > \frac{2\ln k}{2\sqrt{2\ln k}} = \frac{\sqrt{\ln k}}{\sqrt{2}} > \rho.
	\end{equation}
	Similarly, for every $x \leq k$, it holds that for every $y < x$
	\begin{equation}\label{eq:M2}
	t_1 + (x-y)t_{k-y} = \frac{1}{\sqrt{2\ln k}} \left(1 + \frac{x-y}{k-y}\right) \leq \frac{2}{\sqrt{2\ln k}} = \frac{\sqrt{2}}{\sqrt{\ln k}} < \frac{1}{\rho}.
	\end{equation}
	Moreover, we have that for every $x = 2, \ldots, k$
	\begin{equation}\label{eq:no_interl}
	\frac{1}{(x-1)\sqrt{2\ln k}} + (x-1)\frac{1}{x\sqrt{2\ln k}} = \frac{1}{\sqrt{2\ln k}} \frac{x^2 - x + 1}{x^2-x}\leq \frac{1}{\sqrt{\ln k}} \frac{3}{2\sqrt{2}} < \frac{\sqrt{2}}{\sqrt{\ln k}} < \frac{1}{\rho}.
	\end{equation}
	On the other hand, for $x \geq \ln k + 1$ we have
	\begin{equation}\label{eq:no_interl2}
	\frac{1}{x\sqrt{2\ln k}} + (x-1)\frac{1}{\sqrt{2\ln k}} \geq (x-1)\frac{1}{\sqrt{2\ln k}} \geq \frac{\ln k}{\sqrt{2}\sqrt{\ln k}} = \frac{\sqrt{\ln k}}{\sqrt{2}}> \rho.
	\end{equation}
	Finally, observe that
	\begin{equation}\label{eq:sum}
	\sum_{x=1}^k \frac{1}{x \sqrt{2\ln k}} = \frac{H_k}{\sqrt{2\ln k}} \geq \frac{\ln k}{\sqrt{2}\sqrt{\ln k}} > \rho \quad (H_k \text{ denoting the $k$-th armonic number}),
	\end{equation}
	and that
	\begin{equation}\label{eq:red_sum}
	\sum_{x=1}^{2\left\lceil\ln k\right\rceil} \frac{1}{x \sqrt{2\ln k}} = \frac{H_{2\left\lceil\ln k\right\rceil}}{\sqrt{\sqrt{2}\ln k}} \leq \frac{1+\delta}{\sqrt{2}}\frac{\ln \left(2\left\lceil\ln k\right\rceil\right)}{\sqrt{\ln k}} < 1,
	\end{equation}
	where $\delta$ is a small constant that goes to 0 as $k$ increases, and the last inequality holds for $k$ sufficiently large.
	
	Let us rename the agents as $i_k, \ldots, i_1$ where $i_j$ is the agent with the maximum out-priority on the history compatible with $i_{j+1}, \ldots, i_k$ having valuation $t_{j+1}, \ldots, t_k$, respectively. 
	We now focus on the instance $\b$ such that $b_{i_S} = 1$ and $b_{i_k} = t_k$.
	Observe that $S$ cannot be returned on this instance, otherwise by \eqref{eq:sum}, the mechanism would have an approximation ratio worse than $\rho$. Hence, let us now focus only on the rankings of $i \in T$.
	
	We first prove that for every $O \subseteq T$ of size at least $k - 2\lceil\ln k\rceil$ it holds that $g^{(out)}_i(0, \b_{N \setminus {\cal A}}) > g^{(in)}_i(1,\b_{N \setminus {\cal A}})$ for every $i \in O$ and every ${\cal A}$ such that $N \setminus O \subseteq {\cal A}$. Suppose indeed that this is not the case, and there is a set $U$ of size at least $2\lceil\ln k\rceil$ and a suitable ${\cal A}$, such that $g^{(in)}_i(1,\b_{N \setminus {\cal A}}) > g^{(out)}_i(0,\b_{N \setminus {\cal A}})$ for every $i \in U$.
	We must then have that $g^{(out)}_i(0,\b_{N \setminus {\cal A}}) > g^{(in)}_i(t_1,\b_{N \setminus {\cal A}})$ for every $i \in U$, because otherwise the solution $T \cap {\cal A}$ will be returned by the algorithm even on the instance $\b'$ such that $b'_{i_S} = 1$, $\b'_{N \setminus {\cal A}} = \b_{N \setminus {\cal A}}$, $b'_i = \min_{j \in {\cal A}} b_j$ for $i \in {\cal A} \setminus U$, $b'_{i^*} = t_1$, and $b'_i = 0$ for every $i \in U \setminus \{i^*\}$, where $i^*$ is the one $i \in U$ for which $g^{(out)}_i(0,\b_{N \setminus {\cal A}}) < g^{(out)}_i(t_1,\b_{N \setminus {\cal A}})$. Hence, the solution of the returned solution is $t_1 + |{\cal A} \setminus U| \min_{j \in {\cal A}} b_j$, that, by \eqref{eq:M2} is worse than a $\rho$-approximation of the optimal solution (i.e., $S$). Thus, for every $i \in U$ we have
	$
	g^{(in)}_{i}(1,\b_{N \setminus {\cal A}}) > g^{(out)}_{i}(0,\b_{N \setminus {\cal A}}) >  g^{(in)}_{i}\left(t_1,\b_{N \setminus {\cal A}}\right)
	$
	and then a change of direction occurs for all these agents. By the interleaving property, either there is $i$ that is selected (and thus $S$ is not) in both $\b'$ and $\b''$ or no $i \in U$ is selected in $\b''$, where $\b'$ is as defined above, and $b''_{i_S} = 1$, $\b''_{N \setminus {\cal A}} = \b_{N \setminus {\cal A}}$, $b''_j = t_2$ for every remaining $j$ (it includes every $i \in U$). By \eqref{eq:M1} and \eqref{eq:M2}, it turns out that in both cases the approximation factor is worse than $\rho$.
	
	We next show that no change of direction in the priorities of agents $i \in O$ occurs until at least $k-\left\lceil\ln k\right\rceil$ of them have been dropped.
	Consider otherwise that a change of direction occurs when only $x$ agents have been dropped, for $x < k-\left\lceil\ln k\right\rceil$. Take then instances $\b'$ and $\b''$, for which there is a subset $U \subseteq T$ of size $k-x$ such that $b_i' = b''_i = 0$ for every $i \in U$, $b'_{i_S} = b''_{i_S} = 1$, $b'_{i^*} = t_x$, $b'_{j} = t_1$, while $b''_{i^*} = t_{x-1}$, $b''_{j} = t_x$, for all $j \neq i^* \in T \setminus U$. Now by \eqref{eq:no_interl}, $i$ must not be part of the solution for $\b''$ whilst by \eqref{eq:no_interl2} she must be in the solution returned for $\b'$. Therefore, there are two different types $t_{x-1}>t_x$ whose monotone priority cannot be defined when the change of direction occurs. This means that the algorithm is not interleaving, a contradiction.
	
	Putting together the two properties stated above, the algorithm returns on the instance $\b$ a solution comprising of at most $2\left\lceil\ln k\right\rceil$ agents from $T$, chosen among the ones with largest valuation. However, from \eqref{eq:red_sum}, this solution would be worse than returning $S$. Hence, from \eqref{eq:sum} it turns out that the optimal solution would be $T$, and its welfare would be $\rho$ times larger than that of returned solution.
\end{proof}

\section{Conclusions}\label{sec:conclusions}
OSP has attracted lots of interest in the computer science and economics communities, as witnessed by the flurry of work on the topic (see overview in Appendix \ref{apx:rel}). Our work can facilitate the study of this notion of incentive-compatibility for imperfectly rational agents. 

Just as the characterization of DAAs in terms of reverse greedy  \cite{MS20} has given the first extrinsic reason to study the power and limitations of these algorithms \cite{DGR17,GMR17}, we believe that our characterization of OSP in terms of two-way greedy will lead to a better understanding of this algorithmic paradigm. In this work, we only began to investigate their power and much more is left to be done. For example, different optimization problems and objective functions (such as, min-max) could be considered. Moreover, whilst we know that in general OSP mechanisms do not compose sequentially, see, e.g., \cite{badegonczarowski}, we could study under what conditions two-way greedy algorithms compose, just like done for the composability of monotone algorithms for strategyproofness \cite{MN02}.

The technical approach that we have used for our characterization of OSP goes beyond binary allocation problems. Studying the relationship between non-negative short and negative longer cycles will shed light on more general domains/problems, such as, scheduling with single-parameter agents \cite{esa19}. The richer solution space should yield a more intricate structure for negative-weight cycles and eventually lead to a characterization of  implementation trees (and, consequently, algorithms). Importantly, there is evidence that our approach is also useful in the context of multi-parameter agents, as recently shown in \cite{bartmaria} for unknown single-minded combinatorial auctions (where bidders can lie about valuation \textit{and} set of items they are willing the buy). The extent to which this is true more in general is a valuable question to investigate.

The computational complexity of OSP mechanisms is another research direction worthy of mention. The complexity of these mechanisms comes for their three components: algorithm (intuitively, the social choice function run in the leaves of the implementation tree); protocol (i.e., the queries dictated by the implementation tree); and, payments. We know that the algorithms run in time linear to the number of agents, provided that the priority functions can be computed in constant time. This is indeed the case for all the algorithms listed in Table \ref{tbl:upperbounds}. The execution protocol measures a sort of query complexity, which is captured, by means of our characterization, by the sum of the agents' domain sizes (the worst-case being the instance in which each agent has type equal to the last query of the protocol). (Defining the entire protocol might require longer; in principle, we might need the time to explicitly write down all priority functions used by the two-way greedy algorithm.) We believe that an interesting object of study here could be the trade-off between query complexity and approximation guarantee of these mechanisms. Our payment functions are implicitly defined through the cycle-monotonicity framework (as certain shortest paths in the \vgraph s, see \cite{ferraioli2018approximation}); this means that we need time proportional to the size of the \vgraph s for their computation. Another important problem left open by our work concerns in fact the definition of explicit payments functions for the OSP implementation of two-way greedy algorithms. We give a taster of this direction in Appendix \ref{sec:cas} in the context of combinatorial auctions. Together with the impact upon the mechanism running time, this kind of results would also make the meaning of the imperfect rationality behind OSP -- and, in turns, its notion of ``simplicity'' -- more explicit. 

\appendix

\section{Postponed proofs}

\subsection{Proof of Observation~\ref{obs:binary}}
\label{apx:binary}
\begin{proof}
	Let $v$  a node of $\T$ where the domain of $i=i(v)$ is partitioned into $k > 2$ parts $P_1, \ldots, P_k$, and let  $\T_1, \ldots, \T_k$ be the corresponding subtrees rooted at the $k$ children of $v$. 
	We modify $\T$ locally to make it binary in the following way. We substitute $v$ with $k-1$ new nodes, $v^{(1)}, \ldots, v^{(k-1)}$, with $i=i(v^{(j)})$ for  each $v^{(j)}$. At $v^{(1)}$ we ask agent $i$ to separate $P_1$ from $P_2 \cup \cdots \cup P_k$. The left child of $v^{(1)}$ is thus $\T_1$, and the right child is $v^{(2)}$ where we iterate the same approach on $P_2 \cup \cdots \cup P_k$. More generally, at each node $v^{(j)}$, we ask agent $i$ to separate $P_j$ from $P_{j+1} \cup \cdots \cup P_k$,  the subtree rooted at the left child is $\T_j$,  while the right child is $v^{(j+1)}$. It is not hard to see, by inspection, that these modifications do not alter the OSP constraints of $i$ at $v$ in the original mechanism $\M$; similarly, we note that the OSP constraints defined by nodes different from $v$ are not modified. 
	By reiterating the argument on all the $v \in \T$ with more than two children, we can construct a binary tree $\T'$ that guarantees OSP and uses the same $(f,p)$. This concludes the proof.
\end{proof}

\subsection{Local transformations}
\label{apx:local_tr}
Suppose $\T$ contains a node $u$ whose query to $i=i(u)$ is not extremal, meaning that neither  two subdomain $L^{(u)}$ and $R^{(u)}$ is a singleton.  Let $\T_{L^{(u)}}$ and $\T_{R^{(u)}}$ denote the subtrees rooted in the child of $u$ reached through the types in $L^{(u)}$ and $R^{(u)}$, respectively.

We define below a suitable \emph{local transformation} that  replace $u$ with a sequences of nodes whose queries are extremal and that satisfy these two conditions: 
\begin{enumerate}
	\item \label{local:CMON-others} If the original mechanism satisfies OSP CMON for all agents $j$ other than $i$, then so does the modified mechanism; 
	\item \label{local:2CMON-i} If the original mechanism satisfies OSP 2CMON for agent $i$, then so does the modified mechanism.
\end{enumerate}

This transformation proceeds as follows:
Each node $u$ in $\T$ is replaced by a sequence of nodes $u^{(1)}, \ldots, u^{(s)}$, with  $i(u)=i(u^{(1)})=\cdots = i(u^{(s)})$ and each of them having two children;  node $u^{(h)}$ with $h<s$ has $u^{(h+1)}$ as a child, and the other child is either 
\begin{enumerate}
	\item ($L$-sub.) A copy of  $\T_{L^{(u)}}$ whose root we denote by $\ell^{(h)}$. In this case we say that $u^{(h)}$ is a $L$-sub. 
	\item ($R$-sub.) A  copy of  $\T_{R^{(u)}}$ whose root we denote by $r^{(h)}$. In this case we say that $u^{(h)}$ is a $R$-sub. 
\end{enumerate}
 Finally, the last node $u^{(s)}$ has as children a copy of $\T_{L^{(u)}}$ and a copy of $\T_{R^{(u)}}$, thus it is both a $L$-sub and a $R$-sub (with $\ell^{(s)}$ and $r^{(s)}$ denoting the roots of these copies as above).

\newcommand{\rsub}[1]{D_i^{({r^{(#1)}})}}
\newcommand{\lsub}[1]{D_i^{({\ell^{(#1)}})}}
\newcommand{\rootsub}[1]{D_i^{(root(#1))}}

Our aim is to choose nodes and which child to assign to them so that the following property is satisfied.
\begin{definition}[Well-Aligned Transformation]
	A local transformation is \emph{well aligned} if 
	\begin{align*}
	\left(\bigcup_{u^{(h)} \textrm{is a $R$-sub\ }} \rsub{h}\right) = R^{(u)} && \text{ and } && \left(\bigcup_{u^{(h)} \textrm{ is a $L$-sub\ }} \lsub{h}\right) = L^{(u)},
	\end{align*}
\end{definition}
 It is not hard to see that, well aligned local transformations do not destroy  OSP for the agents other than $i$ (Property~\ref{local:CMON-others}).
 We next describe how our (well aligned) local transformation are implemented, and we observe that OSP 2CMON of agent $i$ is also preserved (Property~\ref{local:2CMON-i}).

 \subsubsection{Local transformation for Homogeneous Case}
 We can now define $\T'$ through a local transformation of $u$ where we sequentially ask top queries from $Q=R \cup L^{(\geq r_{\min})}$. (Clearly, here we use $\T_L$ or $\T_R$ as children of the new nodes in $\T'$ depending on whether the current maximum belongs to $L^{(\geq r_{\min})}$ or $R$.)  Observe that, by construction, OSP 2CMON is satisfied for $\T'$, since we never add a negative-weight two-cycle.

\subsubsection{Local transformation for Revealable Case} We consider the types in $L$ in increasing order, that is, $L=\{l_1<l_2<\cdots < l_\ell \}$ for $\ell =|L|$. Then, 
we will define a suitable local transformation which corresponds to a sequence of bottom  queries about all types in $L$. Specifically,
we  apply the following local transformation of node $u$ defined by the sequence of nodes $u^{(1)}, \ldots, u^{(\ell)}$ for $\ell = |L|$. At $u^{(q)}$ we query $i$ to distinguish between $l_q$ and the remaining higher possible values in $\{l_{q+1},\ldots,l_\ell\} \cup R$. 
The child of $u^{(q)}$ associated to $l_q$ is $\T_L$, and the other child  of $u^{(q)}$ is next node $u^{(q+1)}$ of the transformation. Finally, the two children of last node $u^{(\ell)}$ are $\T_L$  for action corresponding to $l_\ell$ (the maximum in $L$), and $\T_R$ for the action corresponding to the types in $R$.  We let $\T'$ denote the tree obtained from $\T$ by applying this local transformation. Observe that, by construction, OSP 2CMON is satisfied for $\T'$, since we never add a negative-weight two-cycle. 

\subsection{Proof of Observation \ref{obs:wellordered}}
	\begin{proof}
	The proof discusses the case in which types are costs (but, as above, does not require assumptions on their signs). Given an OSP extremal mechanism $\M$ in which agent $i$ is revealable at node $u$, we denote with $L$ all types $t \in D_i^{(u)}$ that are $\always{1}$, and with $R$ all the types that are $\always{0}$. 
	
	We build the mechanism $\M'$ by modifying $\T$ to $\T'$ as follows. Assume first that there is type $\hat t$, which is \neither. Note that since $\M$ is OSP 2CMON, then we have that $\max L < \hat{t} < \min R$. We create $d=|D_i^{(u)}|$ copies of $u \in \T$, denoted $u_1,\dots, u_d$, each rooting the subtree of $\T$ rooted at $u$. We also create $d-1$ vertices $v_1, \ldots, v_{d-1}$ in $\T'$ where $i(v_j)=i$ for all $j$. At the $v_j$'s, we make a sequence of top queries about types in $R$, then a sequence of bottom queries for the types in $L$. The children of each $v_j$, $j < d-1$, are the $j$-th copy of $u$, $u_j$, and $v_{j+1}$. The children of $v_{d-1}$ are $u_{d-1}$ and $u_d$, the last two copies of $u$. Note that the last query for $i$ at $v_{d-1}$ will separate the maximum in $L$ from $\hat t$. This query will not be needed if there is no \neither\ type $\hat t$ in $D_i^{(u)}$ and therefore the number of $u$'s and $v$'s in $\T'$ will be reduced by one. 
	
	We now argue that $\M'$ is OSP CMON. By construction, the \vgraph s ${\cal O}_j^{\T'}$ of each agent $j \neq i$ have less edges than ${\cal O}_j^{\T}$ and it does not have negative-weight cycles since $\M$ is OSP. Let us then focus on agent $i$ and its \vgraph\ $\verp$ which has more edges than $\ver$. By construction, these new edges satisfy OSP 2CMON; assume by contradiction that $\verp$ has a negative weight cycle $C$ and let $\b^{(1)}, \b^{(2)}, \b^{(3)}$ and $\b^{(4)}$ the profiles in $C$ defined by Theorem \ref{thm:anatomy}. Recall that there must be no edge between $\b^{(2)}$ and $\b^{(3)}$ in $\verp$, where $b_i^{(2)}<b^{(3)}_i$ and $f_i(\b^{(2)})=0$ and $f_i(\b^{(3)})=1$. It is not hard to see that there are no two values in $D_i^{(u)}$ that satisfy this condition --- a contradiction.
\end{proof}

\section{Payments of IAAs for Single-Minded CAs}\label{sec:cas}
In \cite{MS20} it is observed how \emph{a particular} immediate-acceptance auction based on the forward greedy algorithm, inspired\footnote{The algorithm in \cite{LOS} uses a slightly more articulate priority function to return a better approximation of the optimal social welfare; we here follow \cite{MS20} and focus on the simpler greedy-by-valuation algorithm. It is shown in \cite{bartmaria} how this algorithm can be implemented in an OSP way for unknown single-minded bidders.} by \cite{LOS}, where bidders are ranked by their valuation (i.e., using the notation of Section \ref{sec:2way}, $g_i^{(in)}(b)=b$, for all $i \in N, b \in D_i$) is not group-strategyproof and hence not OSP. The immediate-acceptance auction in \cite{MS20} uses the threshold payment scheme: the winners are charged the minimal bids that would have been accepted \emph{for the particular bid profile of the other bidders}. We now show how Corollary \ref{cor:iaa} applies to this setting for the very simple instance used by \cite{MS20} --- this illustrates how the payment scheme needs to be built upon the implementation tree (rather than blindly use the SP payments for the direct-revelation implementation).

Consider the following simple instance of single-minded CAs. We have three bidders: bidder 1 and bidder 2 want a different object each, whilst bidder 3 wants the bundle. Let the valuation domain of the bidders be $\{\tmin,\tmed,\tmax\}$, with $\tmin=0$ and $2\tmed <\tmax$. Note that the threshold for bidder 1 and 2 for the bid profile $\b_{\max}=(\tmax,\tmax,\tmax)$ is $0$, as either bidder would win even by declaring $\tmin$. (Here we can safely assume that the algorithm prefers the two bidders; technically speaking, the threshold need not be \always{1} -- cf. Definition \ref{def:typeclas}.) Therefore, bidder 1 and 2 would find it beneficial to jointly misreport $\tmed$ to $\tmax$ when bidder $3$ bids $\tmax$ \cite{MS20} as they would win the items for free (whilst they would lose the items if truthful).

How does Corollary \ref{cor:iaa} circumvent this problem? The implementation tree used by the IAA first queries bidder 1 for $\tmax$ (and returns $\{1,2\}$ if yes) then bidder $2$ for $\tmax$ (and again returns $\{1,2\}$ if yes) and finally bidder $3$ for $\tmax$ returning $\{3\}$ if yes; the rest of the tree is immaterial for our purposes. (To break ties and force this query ordering, we could simply modify the in-priority functions $g_i^{(in)}$'s by adding an increasing function of $i^{-1}$.) It is not too hard to see that this implementation tree does not require the payment for bidder 1 and 2 to be the same for the bid profile $\b_{\max}$. In fact, the only OSP-graph with edges from/to $\b_{\max}$ is bidder 1's and so bidder $2$'s payment for that bid profile is not constrained by OSP; in other words, bidder $2$ is not concerned with the outcome of the auction for $\b_{\max}$ when taking her decision about separating $\tmax$ from the rest of her domain. 

\section{Further Related Work}\label{apx:rel}
We will here focus only on the body of work on OSP. Since its inception \cite{liosp}, OSP has attracted the interest of both computer scientists and economists. The authors of \cite{ashlagigonczarowski} study OSP for stable matchings and provide an impossibility result in general as well as a suitable mechanism under some assumptions. The paper \cite{badegonczarowski} further studies this concept for a number of settings, including single-peaked domains. Further work on OSP for single-peaked preferences is \cite{AMN19,AMN20}. In \cite{pycia2}, the authors study OSP mechanisms in domains where monetary transfers are not allowed and provide a useful characterization in this setting. 

For the setting where money is permitted, OSP is investigated for machine scheduling and set system problems in \cite{ferraioli2018approximation}, published in \cite{esa19,wine19}, where mechanisms are designed for ``small'' domains of up to three values, the cycle-monotonicity technique is introduced and some lower bounds are provided. The authors of \cite{ferraioli2017obvious} study OSP under a restriction on the agents' behavior during the execution of the mechanism, called \emph{monitoring}. This is shown to help the design of OSP mechanisms with a good approximation ratio in various mechanism design domains. The paper \cite{kyropoulou2019obviously} builds forth on this by studying machine scheduling in the absence of monetary payments and showing the power of OSP in this setting. Another paper which focuses on machine scheduling is \cite{ferraioliventrebounded}, where a generalization of OSP is studied that allows for some ability of the agents to reason contingently (as opposed to none); the authors show that a large amount of ``lookahead''-ability is required for the agents in any mechanism that achieves a good approximation ratio. Another study that considers OSP under a restriction on the agents' behavior is \cite{verification}, where the authors assume that non-truthful behaviour can be detected and penalised with a certain probability. Under this assumption, the authors prove that every social choice function can be implemented by an OSP mechanism with either very large fines for lying or a large number of ``verified'' agents. A recent paper studies OSP for single-minded combinatorial auctions \cite{bartmaria}; the authors use the structural approach of negative long cycles we introduce herein for the multidimensional setting of unknown single-minded bidders. In \cite{mackenzie} a revelation principle is presented that states that every social choice function implementable through an OSP mechanism can be implemented using a certain structured OSP protocol where agents take turns making announcements about their valuations. This informs our definition of extensive-form mechanisms and OSP. 

\newpage
\bibliographystyle{abbrv}
\bibliography{ospb,osps}

\begin{thebibliography}{10}

\bibitem{AMN19}
R.~P. Arribillaga, J.~Mass{\'o}, and A.~Neme.
\newblock All sequential allotment rules are obviously strategy-proof.
\newblock 2019.

\bibitem{AMN20}
R.~P. Arribillaga, J.~Mass{\'o}, and A.~Neme.
\newblock On obvious strategy-proofness and single-peakedness.
\newblock {\em Journal of Economic Theory}, 2020.

\bibitem{ashlagigonczarowski}
I.~Ashlagi and Y.~A. Gonczarowski.
\newblock Stable matching mechanisms are not obviously strategy-proof.
\newblock {\em Journal of Economic Theory}, 177:405--425, 2018.

\bibitem{ausubel2004}
L.~M. Ausubel.
\newblock An efficient ascending-bid auction for multiple objects.
\newblock {\em American Economic Review}, 94(5):1452--1475, 2004.

\bibitem{Avi83}
D.~Avis.
\newblock A survey of heuristics for the weighted matching problem.
\newblock {\em Networks}, 13(4):475–493, 1983.

\bibitem{badegonczarowski}
S.~Bade and Y.~A. Gonczarowski.
\newblock Gibbard-satterthwaite success stories and obvious strategyproofness.
\newblock In {\em Proceedings of the 2017 ACM Conference on Economics and
  Computation}, EC ’17, page 565, New York, NY, USA, 2017. Association for
  Computing Machinery.

\bibitem{Borodin}
A.~Borodin, M.~Nielsen, and C.~Rackoff.
\newblock ({I}ncremental) priority algorithms.
\newblock {\em Algorithmica}, 37(4):295--326, 2003.

\bibitem{bartmaria}
B.~de~Keijzer, M.~Kyropoulou, and C.~Ventre.
\newblock Obviously strategyproof single-minded combinatorial auctions.
\newblock In {\em ICALP}, pages 71:1--71:17.

\bibitem{DGR17}
P.~D{\"{u}}tting, V.~Gkatzelis, and T.~Roughgarden.
\newblock The performance of deferred-acceptance auctions.
\newblock {\em Math. Oper. Res.}, 42(4):897--914, 2017.

\bibitem{ferraioli2018approximation}
D.~Ferraioli, A.~Meier, P.~Penna, and C.~Ventre.
\newblock On the approximation guarantee of obviously strategyproof mechanisms.
\newblock {\em arXiv preprint arXiv:1805.04190}, 2018.

\bibitem{wine19}
D.~Ferraioli, A.~Meier, P.~Penna, and C.~Ventre.
\newblock Automated optimal {OSP} mechanisms for set systems - the case of
  small domains.
\newblock In {\em Web and Internet Economics - 15th International Conference,
  {WINE} 2019, New York, NY, USA, December 10-12, 2019, Proceedings}, pages
  171--185, 2019.

\bibitem{esa19}
D.~Ferraioli, A.~Meier, P.~Penna, and C.~Ventre.
\newblock Obviously strategyproof mechanisms for machine scheduling.
\newblock In {\em 27th Annual European Symposium on Algorithms, {ESA} 2019,
  September 9-11, 2019, Munich/Garching, Germany}, pages 46:1--46:15, 2019.

\bibitem{ferraioli2017obvious}
D.~Ferraioli and C.~Ventre.
\newblock Obvious strategyproofness needs monitoring for good approximations.
\newblock In {\em Thirty-First AAAI Conference on Artificial Intelligence},
  2017.

\bibitem{verification}
D.~Ferraioli and C.~Ventre.
\newblock Probabilistic verification for obviously strategyproof mechanisms.
\newblock In {\em Proceedings of the Twenty-Seventh International Joint
  Conference on Artificial Intelligence, {IJCAI-18}}, pages 240--246, 2018.

\bibitem{ferraioliventrebounded}
D.~Ferraioli and C.~Ventre.
\newblock Obvious strategyproofness, bounded rationality and approximation.
\newblock In {\em Algorithmic Game Theory}, pages 77--91. Springer, 2019.

\bibitem{GMR17}
V.~Gkatzelis, E.~Markakis, and T.~Roughgarden.
\newblock Deferred-acceptance auctions for multiple levels of service.
\newblock In {\em Proceedings of the 18th Annual ACM Conference on Economics
  and Computation (EC). 21–38.}

\bibitem{Hausmann1980}
D.~Hausmann, B.~Korte, and T.~A. Jenkyns.
\newblock Worst case analysis of greedy type algorithms for independence
  systems.
\newblock {\em Combinatorial Optimization}, pages 120--131, 1980.

\bibitem{kagel87}
J.~H. Kagel, R.~M. Harstad, and D.~Levin.
\newblock Information impact and allocation rules in auctions with affiliated
  private values: A laboratory study.
\newblock {\em Econometrica}, 55(6):1275--1304, 1987.

\bibitem{K56}
J.~B. Kruskal.
\newblock On the shortest spanning subtree of a graph and the traveling
  salesman problem.
\newblock {\em Proceedings of the American Mathematical Society}, 7(1):48–50,
  1956.

\bibitem{kyropoulou2019obviously}
M.~Kyropoulou and C.~Ventre.
\newblock Obviously strategyproof mechanisms without money for scheduling.
\newblock In {\em Proceedings of the 18th International Conference on
  Autonomous Agents and MultiAgent Systems}, pages 1574--1581, 2019.

\bibitem{Cla83}
K.~L.Clarkson.
\newblock A modification of the greedy algorithm for vertex cover.
\newblock {\em Information Processing Letters}, 16(1):23--25, 1983.

\bibitem{LOS}
D.~Lehmann, L.~I. O'Callaghan, and Y.~Shoham.
\newblock Truth revelation in approximately efficient combinatorial auctions.
\newblock {\em J. ACM}, 49(5):577–602, sep 2002.

\bibitem{liosp}
S.~Li.
\newblock Obviously strategy-proof mechanisms.
\newblock {\em American Economic Review}, 107(11):3257--87, November 2017.

\bibitem{mackenzie}
A.~Mackenzie.
\newblock {A revelation principle for obviously strategy-proof implementation}.
\newblock Research Memorandum 014, Maastricht University, Graduate School of
  Business and Economics (GSBE), May 2018.

\bibitem{MS20}
P.~Milgrom and I.~Segal.
\newblock Clock auctions and radio spectrum reallocation.
\newblock {\em Journal of Political Economy}, 2020.

\bibitem{MN02}
A.~Mu'alem and N.~Nisan.
\newblock Truthful approximation mechanisms for restricted combinatorial
  auctions.
\newblock {\em Games Econ. Behav.}, 64(2):612--631, 2008.

\bibitem{book}
N.~Nisan, T.~Roughgarden, E.~Tardos, and V.~Vazirani, editors.
\newblock {\em Algorithmic Game Theory}.
\newblock 2017.

\bibitem{pycia2}
M.~Pycia and P.~Troyan.
\newblock Obvious dominance and random priority.
\newblock In {\em Proceedings of the 2019 ACM Conference on Economics and
  Computation}, EC ’19, page~1, New York, NY, USA, 2019. Association for
  Computing Machinery.

\bibitem{SY05}
M.~E. Saks and L.~Yu.
\newblock Weak monotonicity suffices for truthfulness on convex domains.
\newblock In {\em Proceedings 6th {ACM} Conference on Electronic Commerce
  (EC-2005), Vancouver, BC, Canada, June 5-8, 2005}, pages 286--293, 2005.

\end{thebibliography}

\end{document}